\documentclass{article}
\usepackage{kr}
\usepackage{times}
\usepackage{soul}
\usepackage{url}
\usepackage[hidelinks]{hyperref}
\usepackage[utf8]{inputenc}
\usepackage[small]{caption}
\usepackage{graphicx}
\usepackage{amsmath}
\usepackage{amsthm}
\usepackage{amssymb}
\usepackage{amsfonts}
\usepackage{booktabs}
\usepackage{paralist}
\usepackage{stmaryrd}
\usepackage[vlined,ruled]{algorithm2e}
\usepackage[usenames,dvipsnames,svgnames,table]{xcolor}
\usepackage{tcolorbox}

\usepackage{float}
\restylefloat{table}
\usepackage[usenames,dvipsnames,svgnames,table]{xcolor}
\usepackage{tcolorbox}
\usepackage{tikz}

\usetikzlibrary{matrix,arrows.meta}
\usetikzlibrary{calc,arrows,shapes}
\usetikzlibrary{trees,fit,shapes,calc}
\tikzstyle{mybox} = [rectangle,draw,minimum height=1cm, minimum width=6em,fill=blue!20]
\tikzstyle{decision} = [diamond, draw,
text width=4.5em, text badly centered, node distance=4cm, inner sep=0pt]
\tikzstyle{block} = [rectangle, draw,
text width=5em, text centered, rounded corners, node distance=4cm, minimum height=4em]
\tikzstyle{line} = [draw, -latex']
\tikzstyle{cloud} = [draw, ellipse, node distance=4cm,
minimum height=2em]

\usepackage{mathtools}

\usepackage{cuted}
\usetikzlibrary{arrows.meta,shapes.misc,patterns,shapes.symbols,decorations.pathreplacing,decorations}

\DeclareMathSymbol{\mlq}{\mathord}{operators}{``}
\DeclareMathSymbol{\mrq}{\mathord}{operators}{`'}

\usepackage{tikz}
\usetikzlibrary{trees}

\tikzset{My Style/.style={blue, draw=black,fill=pink, minimum size=0.5cm}}

\newcommand{\ignore}[1]{} 

\newcommand{\defin}[1]{\left\{\begin{array}{l} #1 \end{array}\right\}}

\newcommand{\rul}{\leftarrowtail}

\newtheorem{definition}{Definition}
\newtheorem{theorem}{Theorem}
\newtheorem{lemma}{Lemma}

\newtheorem{example}{Example}

\ignore{

}

\newcommand{\bX}{\bar{X}}  

\newcommand{\bZ}{\bar{Z}}  
  
\newcommand{\bY}{\bar{Y}}  
\newcommand{\bP}{\bar{P}}

\newcommand{\ba}{\bar{a}}

\newcommand{\eq}{=}

\newcommand{\logic}{{\rm DetLDL}$_f$\,{\rm[SM-PP]}\xspace}

\newcommand{\CPT}{$\tilde{\rm C}${\rm PT(Card)}\xspace}

\newcommand{\e}{\varepsilon}

\newcommand{\GuessNewO}{\textit{GuessNewO}}
\newcommand{\GuessNewE}{\textit{GuessNewE}}
\newcommand{\CopyPO}{\textit{CopyPO}}
\newcommand{\CopyPQ}{\textit{CopyPQ}}
\newcommand{\CopyPE}{\textit{CopyPE}}
\newcommand{\GuessP}{\textit{GuessP}}
\newcommand{\GuessNewP}{\textit{GuessNewP}}
\newcommand{\GuessNewPair}{\textit{GuessNewPair}}
\newcommand{\Copy}{\textit{Copy}}
\newcommand{\Reach}{\textit{Reach}}
\newcommand{\Pick}{\textit{Pick}}
\newcommand{\PickPQ}{\textit{PickPQ}}
\newcommand{\Eq}{\textit{Eq}}
\newcommand{\Elim}{\textit{Elim}}
\newcommand{\Root}{\textit{Root}}

\newcommand{\BG}{\textbf{BG}}
\newcommand{\BE}{\textbf{BE}}
\newcommand{\G}{\textbf{G}}
\newcommand{\E}{\textbf{E}}

\newcommand{\Lo}{\mathbb{L}}

\newcommand{\Process}{{\it Process}}
\newcommand{\ChoosenProcess}{{\it Choosen
		Process}}
\newcommand{\ChooseProcess}{{\it ChooseProcess}}
\newcommand{\Union}{{\it Union}}
\newcommand{\ExecuteProcess}{{\it ExecuteProcess}}

\newcommand{\Dom}{{\rm Dom}}

\newcommand{\id}{{\rm id}}

\newcommand{\Un}{{\bf U}}
\newcommand{\UnT}{{\bf U}}
\newcommand{\mT}{{\bf T}}  

\newcommand{\rneg}{\curvearrowright\xspace}
\newcommand{\lneg}{\curvearrowleft\xspace}

\newcommand{\converse}{\smallsmile}

\newcommand{\iter}{\uparrow}
\newcommand{\comp}{\ensuremath \mathbin{;}}

\newcommand{\ppath}{\textbf{p}}

\newcommand{\nat}{\mathbb{N}}  

\newcommand{\strA}{\fA}
\newcommand{\strB}{\fB}         
\newcommand{\strC}{\fC}

\newcommand{\CH}{{\sf CH} }

\newcommand{\SAT}{{\sf SAT} }
\newcommand{\SEN}{{\sf SEN} }

\newcommand{\cL}{\mathcal{L}}
\newcommand{\cP}{\mathcal{P}}

\newcommand{\cA}{\mathcal{A}}

\newcommand{\cR}{\mathcal{R}}

\newcommand{\fA}{\mathfrak{A}}
\newcommand{\fB}{\mathfrak{B}}

\newcommand{\fC}{\mathfrak{C}}

\newcommand{\cK}{\mathcal{K}}

\newcommand{\arity}{\mathrm{ar}}

\newcommand{\ar}{\arity{}}
\newcommand{\iar}{\mathrm{iar}}

\newcommand{\Tr}{\mathit{Tr}}
\newcommand{\last}{\mathit{last}}
\newcommand{\Domain}{\mathcal{D}}
\newcommand{\Last}{\mathsf{Last}}

\newcommand{\emp}{\mathbf{e}}

\newcommand{\schema}{vocabulary\xspace}
\newcommand{\withconverse}[1]{}

\newcommand{\V}{\mathbb{V}}

\newcommand{\Sch}{\mathcal{M}}

\newcommand{\names}{\mathit{Names}}
\newcommand{\vuni}{\V}

\newcommand{\inst}{\ensuremath{T}\xspace}

\newcommand{\semm}[3]{\llbracket #1 \rrbracket_{#2}^{#3}}
\newcommand{\sem}[1]{\semm{#1}{\inst{}}{}}

\newcommand{\ch}{h}

\newcommand{\sema}[1]{\semm{#1}{\inst{}}{\ch}}

\newcommand{\semNoI}[1]{\semm{#1}{}{}}

\newcommand{\bbar}{\text{DR}\xspace}

\newcommand{\gbrv}{global DR\xspace}

\renewcommand{\setminus}{-}

\title{Towards Capturing PTIME with no Counting Construct\\
(but with a Choice Operator)\\
\begin{small}Technical Report of August 2021\end{small}} %

\author{
Eugenia Ternovska
\affiliations
Simon Fraser University
\emails
ter@sfu.ca
}

\begin{document}

\maketitle

\begin{abstract}

The central open question in Descriptive Complexity is whether there is a logic that characterizes deterministic polynomial time (PTIME) on relational structures. Towards this goal, we define a logic that is obtained from first-order logic with fixed points, FO(FP), by a series of transformations that include restricting logical connectives and adding a dynamic version of Hilbert's Choice operator Epsilon. The formalism can be viewed, simultaneously, as an algebra of binary relations and as a linear-time modal dynamic logic, where algebraic expressions describing ``proofs'' or ``programs'' appear inside the modalities. 

We show how counting, reachability and ``mixed'' examples (that include linear equations modulo two) are axiomatized in the logic, and how an arbitrary PTIME Turing machine can be encoded. For each fixed Choice function, the data complexity of model checking is in PTIME. However, there can be exponentially many such functions. 

A crucial question is under what syntactic conditions on algebraic terms checking just one Choice function is sufficient. Answering this question requires a study of symmetries among computations.
This paper sets mathematical foundations towards such a study via algebraic and automata-theoretic techniques.

\end{abstract}

\section{Introduction}\label{sec:Introduction}

The goal of the field of descriptive complexity is to characterize computational complexity in a machine-independent way, in order to apply logical and model-theoretic methods to the study of complexity \cite{Immerman-book}.   However, this goal presents multiple challenges. 
A logical characterization of NP is given by Fagin. His celebrated  theorem \cite{Fagin:1974} states that the complexity class NP coincides, in a precise sense, with second-order existential logic.
The central open question in the area is whether there is a logic that exactly
characterizes deterministic polynomial-time  (PTIME) computability on relational structures.
The problem was first formulated by Chandra and Harel \cite{ChandraHarel:1982} and made precise by Gurevich \cite{Gurevich-challenge} (see also \cite{Grohe:LogicPTIME:2008}). 
Following the formulations of Dawar \cite{Dawar:Gurevich:70}, we have: 
\begin{definition}\label{def:logic} A logic $L$ is a function $\SEN$ associating a recursive set of sentences to each finite vocabulary $\tau$ together with a function $\SAT$ that associates to each $\tau$ a recursive satisfaction relation relating finite $\tau$-structures to sentences that is also isomorphism-invariant. That is, if $\strA$ and $\strB$ are isomorphic $\tau$-structures and $\phi$ is any sentence of $\SEN(\tau)$ then $(\strA, \phi) \in \SAT(\tau)$ if, and only if, $(\strB, \phi) \in \SAT(\tau)$. 
\end{definition}	
		
\begin{definition}\label{def:logic-captures-PRTIME} A logic $L$ \emph{captures PTIME} if 
\begin{enumerate}[(i)]	
\item	there is a computable function that takes each sentence of $L$ to a polynomial  time  Turing machine that recognises the models of the sentence, and 
\item for every polynomial-time recognizable class $\cK$ of structures, there is a sentence of $L$ whose models are exactly $\cK$. 
\end{enumerate}
\end{definition}

Gurevich, in \cite{Gurevich-challenge}, conjectured that a logic for PTIME does not exist. If the conjecture is true, it immediately follows that PTIME $\neq$ NP because,  by Fagin's theorem, \emph{there is} a logic for NP.  Blass, Gurevich and Shelah in \cite{BlassGurevichShelahChoicelessPTIME}  also suggested a formalism of Choiceless Polynomial Time equipped with a cardinality operator (denoted \CPT), in order  to capture the \emph{choiceless} fragment of PTIME.
But, the exact complexity of \CPT is still unknown. Currently, PTIME is known to be captured on linearly ordered structures only. The result was independently obtained by  Immerman and Vardi in the 1980's. It states that, on \emph{linearly ordered} structures, PTIME is captured by 
first-order logic with fixed points, FO(FP)  \cite{Immerman86,Vardi82}. 
However, even extended with counting terms, FO(FP) + C, fixed-point logic is too weak to express all polynomial-time properties on  structures without such order, as shown by 
Cai, F\"{u}rer and Immerman in \cite{CFI92}. 
For \emph{arbitrary}  structures, capturing PTIME remains an intriguing open problem, that, so far, has resisted any attempt to solve it. 

There has been a lot of development in the area.
Dawar, Grohe, Holm and
 Laubner introduced Rank logic in \cite{DawarGHL-rank-logic-LICS:2009}. 
More recent work includes that by Pakusa,
 Schalth{\"{o}}fer and
Selman \cite{PakusaSS18},
Gr\"{a}del and Pakusa \cite{GradelP19},
Atserias, Dawar and
 Ochremiak \cite{AtseriasDO19},
 Dawar and   Wilsenach  \cite{DawarW20}, 
 among others.
The recent result by Lichter separated Rank logic \cite{DawarGHL-rank-logic-LICS:2009} from PTIME \cite{Lichter21}.

Despite the resent advances, Choiceless Polynomial Time, \CPT, proposed in 1999, remains perhaps the most prominent (and the most studied)
candidate for a logic for PTIME. It is based on a machine model of Abstract State Machines applied in a set-theoretic context. 
The idea is to replace arbitrary choice  with a polynomially bounded amount of parallel computation. The  model  uses manipulations of hereditarily finite sets over the input
structure by means of set-theoretic operations. 
A cardinality construct, Card,  associates, with every set, the finite ordinal describing its cardinality.  An equivalent simplified definition of \CPT is
given by Rossman  \cite{Rossman:CPT}. 
Another equivalent formalization is  the polynomial-time interpretation logic (PIL) of \cite{GraedelPakusa-et-al-15} that
``replaces the machinery
of hereditarily finite sets and comprehension terms by traditional
first-order interpretations.'' The goal of introducing this version was to simplify the original formalization and to make typical methods used in finite model theory, for instance	those based on Ehrenfeucht-Fra\"{i}ss\'{e} games, applicable for the analysis of the expressive power.

It is known that \CPT   is strictly more expressive than FO(FP) + C. In particular, it can express the CFI example \cite{CFI92}.  Importantly, in both PIL and  \CPT, 
the built-in cardinality constructs are essential. Without them, it is impossible to express even elementary counting properties such as the query EVEN, as shown by Blass, Gurevich and Shelah  in \cite{BGS:2002}.

In contrast to choiceless computations of \CPT, 
an important use of choice in Descriptive Complexity is that of Gire and Hoang \cite{Gire-Hoang}.
The authors use a choice operator constrained by a formula whose application is restricted  to symmetric elements only.
It extends the power of FO(FP) +  C, while still defining only PTIME properties
(see \cite{Richerby-thesis,Dawar-Richerby} for a detailed study of this formalism). Gire and Hoang’s operator appears to be hard to work with, and it is not easy to prove that a property is inexpressible in their logic.  Similarly to \CPT, it remains open whether this logic captures  PTIME computations on all classes of structures.  
See  \cite{Richerby-thesis} for a detailed discussion.

To achieve further progress in Descriptive Complexity, we believe  it is crucial to make use of  methods from other areas of mathematics.
For instance,  in the area of Constraint Satisfaction Problem,  algebraic techniques have proven to be very effective in the dichotomy results of Bulatov and Zhuk  \cite{Bulatov17,Zhuk17}.
On the other hand,  it is not immediately clear how to represent the cardinality construct of \CPT  and its polynomial restriction on the  length of a parallel computation, algebraically. An equivalent approach of PIL  \cite{GraedelPakusa-et-al-15}, despite making FMT techniques more accessible, has not been shown to make a bridge to other areas of mathematics either. Moreover, PIL uses relations that change their arity from state to state, which are also  difficult to formalize algebraically. 
Similarly, the work of Gire and Hoang  \cite{Gire-Hoang} is far from an algebraic formalization.  The approach requires automorphism tests that are feasible only on  a simple kind of structures. The authors study these structures first, and then  introduce a logical reduction operator that reduces other PTIME queries to queries on such abstract structures. We do not see how to represent  these reduction, or the more precise reformulation by means of first-order interpretations developed by Richerby \cite{Richerby-thesis}, in a simple way suitable for further analysis.

\vspace{2ex}
\noindent\textbf{In this paper},  we develop an algebraic approach from the very beginning. Our goal is to  set foundations for a mathematical study  of the combinatorics of proofs.  
Such  proofs or programs are specified by algebraic terms. The terms are interpreted, first, as binary relations over  relational  structures, and, second, as mappings on finite strings (here, called \emph{paths}) of structures over the same fixed relational vocabulary. Moreover, we utilize a link to a Dynamic Logic to associate, with each algebraic term, a formal language.

The technical ingredients of our proposal are: 
\begin{enumerate}
	\item  an algebra of binary relations on relational structures;
	\item a restriction of logical connectives and an iterator construct to their deterministic (i.e., functional) counterparts;
	\item   access to data values via  equality checks only;
	\item atomic transitions are a restriction of Conjunctive Queries;
	\item  singleton-set monadic relations that change from state to state, while input relations are arbitrary; 
	\item a Choice  operator $\e$, which is  a free function variable ranging over history-dependent concrete Choice functions.
\end{enumerate}

Our general strategy is, first, to allow the logic to express all of NP (due to the use of the free Choice function, which is similar to existential second-order quantification). A query with a free Choice function variable defines a set of concrete Choice functions constrained by an algebraic term. The functions, intuitively, represent proof strategies or  traces of non-deterministic programs. Then, the goal is to use algebraic techniques to study the structure of this set, to find syntactic conditions for a logic for PTIME.

Our technical approach is based on
adding the ability to reason about information propagations in
first-order logic with fixed points, FO(FP). Two parameters are used
to control expressiveness and complexity: the set of 
connectives, which we restrict, and the expressiveness of atomic units, which we take to be a singleton-set monadic modification of conjunctive queries.

\vspace{2ex}

We now briefly explain how the logic was obtained. We took FO(FP) and 
looked into the ``internals'' of declarative computations. By specifying inputs and outputs of atomic expressions,
we, essentially,  developed an ``operational semantics'' of classical logic, expressed as an algebra of binary relations.\footnote{Our motivation, in part, came from the fact that FO$^2$, first-order logic of two variables, has many important decidable properties such as order-invariance \cite{DBLP:conf/lics/ZeumeH16}.  However, for our purposes, we had to go further, from a two-variable logic  to an algebra of binary relations on structures.}
We analyzed the ways information propagates from inputs to outputs, and found a way to control such propagations.  We then investigated what syntactic features make information flows efficient. This investigation was the most significant part of the mental effort behind this project. 

In the resulting algebra, all logical connectives are restricted  to be function-preserving, in the sense that \emph{if} all atomic relations are functional, then so is the compound expression. 
To resolve the remaining non-determinism of atomic expressions,  a  free Choice function variable $\e$ is used. The variable  ranges over  concrete Choice functions that map history to the next state.

The same logic can be viewed in two equivalent ways:
(1) as an algebra of binary relations and  (2)
as a two-way Deterministic Linear Dynamic Logic over finite traces.
The second view provides a link to formal language theory and automata via an equivalent Path semantics. Each path has a natural correspondence to an equivalence class of concrete Choice functions.	 
The high-level idea is to (1)
allow potentially exponentially many ``guessed'' PTIME computations constrained by a formula; (2) use the two views, algebraic and language-theoretic,
to study  symmetries among those computations.
The logic has \emph{no counting construct} and \emph{no programming constructs} --- they are definable.

It is known that adding more and more powerful counting constructs to a logic does not allow one to express reachability. Conversely, logics with extremely powerful iteration mechanism are unable to count, see \cite{Libkin-book04}. A combination of these two fundamentally different types of information propagation is possible  with the help of our Epsilon operator. We use this operator to develop both counting and reachability examples. 
Moreover,  example that  \emph{combine} the two types of propagation, e.g., mod 2 linear equations (that include the CFI example \cite{CFI92}, see \cite{AtseriasBD09}), are also expressible. At the same time, we show that:

\begin{itemize}

	\item  The logic  encodes  an arbitrary PTIME Turing machine; 
	\item For each \textbf{fixed} Choice function, the data complexity of model checking is in PTIME.
\end{itemize} 	
However, there can be exponentially many such functions. 	
The remaining goal is to identify  condition for when all  computations, that correspond to different Choice functions, are indistinguishable in a precise mathematical sense. Decidability of these conditions would give us a logic for PTIME. The syntax of this logic would include the syntax presented in this paper, plus some decidable conditions for  algebraic terms to be ``good''.\footnote{We have already made a significant progress in this direction.}

\vspace{2ex}

\noindent\textbf{The rest of the paper is organized as follows.} In the next section, we discuss related work. 

Then, in \textbf{Section \ref{sec:Preliminaries}}, we introduce basic notations and give technical preliminaries.

In \textbf{Section \ref{sec:Algebra-syntax}}, we present the syntax of the algebra of binary relation, the main algebra we develop in this paper. Moreover, we introduce a particular kind of atomic transitions --- singleton-set monadic, primitive positive definable (SM-PP). The results of this paper concern this algebra, with this type of atomic transitions.

In \textbf{Section \ref{sec:Algebra-semantics}}, we define the semantics of the algebraic operations in terms of binary relations on structures over a fixed relational vocabulary $\tau$. 
First, we define the semantics of a so-called Function-Preserving Fragment, where if atomic relations are functional, then so are the compound expressions. Second, we define the semantics of the Epsilon operator and Converse. Finally, a binary satisfaction relation is given.

Next, in \textbf{Section \ref{sec:Dynamic-Logic}}, we present an equivalent reformulation of the algebra as a modal Two-Way Deterministic Linear Dynamic  Logic (DetLDL$_f$), where proofs or programs appear inside (definable) modalities. As is typical in Dynamic logics \cite{HarelKozenTiuryn-DL}, the syntax is two-sorted,  with one sort for process expressions, that are expressions in the original algebra, and another sort for state formulae, that are evaluated with respect to states (relational structures) of a transition system (Kripke structure). In the same section, we introduce an important boolean query that is true in a structure if and only if there is a successful execution of a given program, that  is a part of the query, on that input structure. Using this query, we define the notion of a computational problem specified by an algebraic term (as an isomorphism-closed class of structures  that make the query true). We also define a set of \emph{witnesses for a computational problem} as a set of Choice functions. Equivalently, a set of witnesses is defined as a set of traces (finite words in the alphabet of structures). The main goal of this paper is to provide tools for the study of such sets of witnesses.
In the same section, we demonstrate that the operations of the algebra are sufficient to define the main imperative programming constructs, so that the full power of Regular Expressions, for that purpose, is not needed.

In \textbf{Section \ref{sec:complexity-main-task}}, we show that the data complexity of checking that there is a successful execution of a program represented by an algebraic term  instantiated with a fixed Choice function,  is in PTIME in the size of the input structure.

In \textbf{Section \ref{sec:Path-Semantics}}, we define a Path semantics of the Dynamic logic DetLDL$_f$.  We also  prove the equivalence of the Path semantics and the semantics in terms of binary relations on structures. This correspondence gives us two interconnected perspectives on the logic --- algebraic and automata-theoretic.  

In  \textbf{Section \ref{sec:Examples}}, we give examples of expressing properties in the logic, including counting and reachability.  We also show how to mimic the constructs of Regular Expressions.

In \textbf{Section \ref{sec:PTIME-TM}}, we show how any PTIME Turing machine can be axiomatized in the logic. The algebraic term, first, guesses an order on a tape of a such a machine;  second, it encodes an input structure on a tape for that order; and, third, it simulates the machine's instructions, until an accepting state is reached.

We conclude, in \textbf{Section \ref{sec:Discussion}}, with a summary, and an outlook for further development.  Moreover, we discuss future research directions that are both  related to the main goal and are of an independent interest.

\vspace{2ex}

\section{Related Work}\label{sec:Related}

We now present a brief comparison with related work.

\noindent\textbf{Algebras of binary relations} have been studied  before.
Such an algebra was first introduced by De Morgan.
It has been extensively developed by Peirce and then Schr\"{o}der. It was abstracted to relation algebra RA by J\'{o}nsson and Tarski in \cite{Jonsson-Tarski:1952}.
More recently, relation algebras were studied by Fletcher,  Van den Bussche,
Surinx and their collaborators in a series of paper, see, e.g. \cite{SVV17,FGLSVVVW}.
The algebras of relations consider various subsets of operations on binary relations as primitive, and other as derivable. 
In another direction, \cite{JacksonStokes:2011,McLean17}  and others study partial functions and their algebraic equational axiomatizations. 

\noindent\textbf{Cylindric Algebras}
Since, in our algebra, all the variables that are not constrained by the algebraic expressions are  implicitly cylinrified,  a closely related work is that on cylindric algebras \cite{HenkinMonkTarskiCylindricAlgebrasPart1}. These algebras were introduced by Tarski and others as a tool in the algebraization of the first-order predicate calculus. See \cite{DBLP:conf/csl/Bussche01} for a historic context in connection with Database theory and Codd's relational algebra. However, a fundamental difference is that, in our logic, unconstrained variables are not only cylindrified, but their interpretation, if not modified, is ``transferred'' to the next state unchanged. We call this property the Law of Inertia.  This property gives us, mathematically, a very different formalism.

\noindent\textbf{Approaches to Choice}, in connection to Finite Model Theory,   
 include the work by Arvind and Biswas \cite{ArvindBiswas87}, Gire and Hoang \cite{Gire-Hoang},  Blass and Gurevich \cite{BlassGurevich:Choice} and by  Otto \cite{DBLP:journals/jsyml/Otto00}, among others.  
  Richerby and Dawar \cite{Richerby-thesis,Dawar-Richerby}  survey  and study logics with various forms of choice, dividing them into three broad categories --- data, formula and denotation choice.  

Choice operator Epsilon was introduced by Hilbert and Bernays in 1939  \cite{HilbertBernays:1939} for proof-theoretic purposes,  without  giving a model-theoretic semantics. The goal was  to systematically eliminate quantifiers in proofs, to reduce them to reasoning in simpler systems. 
Hilbert's $\e$ has been studied extensively by  Soviet logicians Mints, Smirnov and Dragalin in 1970's, 80's and 90's, see \cite{Soloviev17}. The semantics of this operator is still an active research area, see, e.g. \cite{Abrusci17,Wirth17b} and other papers in that journal special issue.\footnote{Journal of Applied Logics -- IFCoLog Journal of Logics and their Applications, Special Issue on Hilbert’s epsilon and tau in Logic, Informatics and Linguistics,
Guest Editors, Stergios Chatzikyriakidis, Fabio Pasquali and Christian Retor\'{e}.
Volume 4, Number 2, March 2017. }

Similarly to \cite{Gire-Hoang}, our Choice is constrained by expressions in the logic. 
But, unlike other approaches, our version of this operator formalizes a \emph{proof strategy}, where, in making a choice, the history is taken into account. 
This dependency on the history reflects  how proofs are constructed in, e.g., Gentzen's proof system --- when selecting an element witnessing an existential quantifier, we must ensure that the element is ``new'', i.e.,  has not appeared earlier in the proof. Given this observation, we believe that proof's history must have been taken into account in Hilbert's reasoning one way or another.\footnote{We find it fascinating that the famous Hilbert's 24th problem concerns the question of when proofs are ``the same''. This question, at least in principle, echoes the choice invariance problem in our formalization.}

\noindent\textbf{Choice Invariance}  One problem with a set-theoretic Choice operator, as discussed in \cite{BlassGurevich:Choice}, is that for FO, and thus for its extensions such as FO(FP), choice-invariance is undecidable. Therefore, it appears that, by restricting our attention to choice-invariant formulae, one would obtain an undecidable syntax, thus violating a basic principle behind a well-defined logic for PTIME.  
We, on the other hand,  
\emph{restrict FO connectives}, similarly to Description logics \cite{DescrLogic2003handbook}. 

\noindent\textbf{Choice vs Choiceless} 
\CPT, according to the authors,  is an attempt to characterize what can be done without introducing choice: ``it is impossible
to implement statements like ``pick an arbitrary element $x$
and continue'', which occur in many high-level descriptions
of polynomial-time algorithms (e.g. Gaussian elimination, the
Blossom algorithm for maximum matchings, and so on)'' \cite{BlassGurevichShelahChoicelessPTIME}.

Our logic, on the other hand, is designed to model precisely such algorithmic constructions, as shown in the examples. In fact, the logic models every polynomial time computation by, first, picking an arbitrary new domain element repeatedly, to guess an order on a tape of a Turing machine; and, then, continuing with the computation for that order.

\noindent\textbf{Linear Order}
Definability of a linear order of domain elements is a common theme in Descriptive Complexity. For the logic of Gire and Hoang \cite{Gire-Hoang}, according to Richerby \cite{Richerby-thesis}, it is unclear whether a linear order is always definable,  and what the power of the logic is on classes of structures where no order can be constructed.  
In \CPT, hereditarily finite sets are used to create all linear orders in parallel, on a small subset of the input structure.

We do not create all orders in parallel, but, in simulating a PTIME Turing machine, a guessed  Choice function corresponds to one such order. The order on domain elements is not defined with respect to the input structure, but is \emph{induced} by an order of states in a Kripke structure that are visited during a computation.

 \noindent\textbf{Other Similarities and Differences}	Similarly to \CPT and PIL, our logic directly  works on structures and uses, essentially, a machine model. In our case, each atomic machine instruction specifies a (non-deterministic) expansion of its inputs with the output relations, while preserving the rest. Unlike other approaches, we have an algebra of binary relations.
It is interesting that  hereditarily finite sets, used in \cite{BGS:2002}, do reappear in our constructions (when we study the combinatorics of proofs), however we manage to avoid working with them explicitly.

 \noindent\textbf{Counting}
 Blass, Gurevich and Shelah in  \cite{BGS:2002} write: 
``What's so special about counting? There is a psychological reason for wanting to add it to any complexity class that doesn't contain it:
Counting is such a fundamental part of our thinking process that its absence in a computation model strikes us as a glaring deficiency.''   
The authors of \cite{GraedelPakusa-et-al-15} also use counting, in the form of H\"{a}rtig
quantifiers (which are classical quantifiers for cardinality comparison).
 
We, on the other hand, do not include any cardinality check as a primitive. We observe that, algorithmically, counting is similar to many other step-wise algorithms. 
Unlike the logics with various forms of cardinality constructs, we formalize algorithmic instructions we could give to a child: e.g., to count to four,  select a random new element exactly four times. 
Needless to say, this process crucially relies on the history-dependent Choice operator.

We also note that in Gire and Hoang's work \cite{Gire-Hoang},  a counting construct is not present, yet the logic contains FO(FP) + C (see also Theorem 37 of \cite{Richerby-thesis} for a more transparent proof).

 \noindent\textbf{Polynomial  Bounds on the Length of a Computation} \CPT \cite{BGS:2002} sets polynomial bounds  on the number of steps and the number of elements a program is allowed to access. These conditions ensure a polynomial upper bound, in terms of the size of the input structure. 
Similar conditions are needed for the polynomial-time interpretation logic \cite{GraedelPakusa-et-al-15}.

In our case, the length of  computations is controlled via internal logical means. States of the transition system are represented by relational structures over the same fixed vocabulary $\tau$. 
The number of tuples  (of singleton-set relations) of a fixed length $k$ is $n^k$. Moreover, the semantics of the algebraic operations excludes infinite loops (if there is a loop, there is no model). Thus, for a fixed Choice function, there could be only a polynomial number of transitions. For a formal proof of this property, please see Theorem \ref{th:membership-PTIME}.

\noindent\textbf{Arities of the Relations} The authors of   \cite{GraedelPakusa-et-al-15}    
state that ``additional
power of \CPT over FO(FP) + C comes from
the generalization of relational iteration in a fixed arity (as in
fixed-point logics) to iterations of relations of changing arities.''
Moreover, according to the authors, one
can easily show that signatures consisting only of monadic	predicates are not sufficient.
	
In contrast, in our logic, all relations that change from state to state are monadic. But, our logic is different from MSO because it can express the query EVEN, which is not in MSO, unless the domain is ordered. We are not aware of any other approach to a logical characterization of PTIME with monadic relations.

\vspace{2ex}

\section{Technical Preliminaries}
\label{sec:Preliminaries}

In this section, we set basic notations and terminology. Most notably, 
a module vocabulary provides names for atomic transductions; instantiation function maps relational variables to constant symbols of a relational vocabulary. We review Conjunctive Queries and PP-definable relations. 
 Importantly, we introduce  Choice functions,  branches and paths.

\noindent A \emph{(module) \schema} $\Sch{}$ is a triple $(\names,\ar{},\iar{})$:
\begin{compactitem}
	\item $\names$ is a nonempty set, the elements of which are called \emph{module names};
	\item $\ar{}$ assigns an arity to each module name in $\names{}$; 
	\item $\iar{}$ assigns an input arity to each module name $M$ in $\names{}$, where $\iar{(M)}\leq\ar{(M)}$.
\end{compactitem}

The terminology and notations in the above paragraph are adopted from Aamer et al., \cite{ABSTV:KR:2020}.
In this paper, $\names$ will be a set of macros that stand for axiomatizations of atomic transductions in some logic, to be specified. 
We assume familiarity with 
the basic notions of first-order (FO) and second-order (SO) logic  (see, e.g., \cite{Enderton}) and use `$:=$' to mean ``is by definition''.

We fix a countably infinite supply of relational variables $\vuni$. 
An \emph{instantiation function} $s : \vuni \to \tau$ maps variables in $\vuni$ to  a relational vocabulary
$\tau$. 
In practice, $\tau$ is a set of relational symbols that are substituted for  variables in a finite expression, so  we  may assume it is finite (but is of an unlimited size).  
Let $\tau \ := \ \{S_1, \dots , S_n\}$,  each $S_i$ has an associated arity $r_i$,  and $\Domain$ be a non-empty set.
A \emph{$\tau$-structure} $\strA$ over  domain $\Domain$  is 
$
\strA \ := \ (\Domain;\  S_1^{\strA}, \dots , S_n^{\strA} ),
$
where $S_i^{\strA}$ is an $r_i$-ary relation  called the \emph{interpretation} of  $S_i$.
If $\strA$ is a $\tau$-structure,  $\strA|_{\sigma}$ denotes
its restriction to a  sub-vocabulary $\sigma$. We use  $\Un$ to denote the set of all
$\tau$-structures over the same domain $\Domain$.

Let $C$ be a class of first-order structures of some vocabulary $\tau$. An $r$-ary \emph{$C$-global relation $Q$} (a \emph{query}) assigns to each structure $\strA$ in $C$ an $r$-ary
	relation $Q(\strA)$ on $\strA$; the relation $Q(\strA)$ is the specialization of $Q$ to $\strA$. The vocabulary $\tau$ is
	the vocabulary  of $Q$. If $C$ is the class of all  $\tau$-structures, we say that $Q$ is
	\emph{$\tau$-global}
\cite{Gurevich-algebra-feasible,Gurevich-challenge}.
Let $\mathcal{L}$ be a logic. A $k$-ary query $Q$ on $C$ is $\mathcal{L}$-definable if there is
an $\mathcal{L}$-formula $\psi(x_1, \dots , x_k)$ with $x_1, \dots , x_k$ as
free variables  such that for every $\strA \in  C$
\begin{equation}\label{eq:query-def}
Q(\strA) = \{(a_1, \dots , a_k) \in \Domain^k \mid  \strA \models \psi(a_1, \dots , a_k)\}.
\end{equation}
Query $Q$ is \emph{monadic} if $k=1$.
A \emph{conjunctive query} (CQ) is a query definable by a first-order formula in prenex normal form built
from atomic formulas, $\land$, and $\exists$ only. A relation is   \emph{Primitive Positive (PP) definable} if it is definable by a CQ:
$\forall x_1 \dots  \forall x_k \ \big( R(x_1, \dots  , x_k)  \leftrightarrow  \exists z_1 \dots \exists z_m\ (B_1 \land \cdots \land B_m)\big),$
and each $B_i$ has object variables from 
$x_1, \dots  , x_k, z_1, \dots , z_m$.
\begin{example}\label{ex:Path2}
{\rm	\textbf{Path of Length Two} is PP-definable:
\begin{equation}\label{eq:Path2}
\forall x_1 \forall x_2 \ \big(Z(x_1, x_2) \leftrightarrow \exists z (Y(x_1, z) \land Y(z, x_2))\big).
\end{equation}
} 
\end{example}
\begin{example}\label{ex:Reach2}
	{\rm Relation	\textbf{At Distance Two  from $X$} is monadic and PP-definable:
		\begin{equation}\label{eq:Reach2}
			\forall x_1 \forall x_2 \ \big(Z(x_2) \leftrightarrow \exists z ( X(x_1)\land Y(x_1, z) \land Y(z, x_2))\big).
		\end{equation}
	} 
\end{example}

Let $\Sigma$ be an alphabet.  
A \emph{finite word} is a finite sequence $w: =a_0\cdot a_1\dots a_n$ of letters in $\Sigma$, where the indexes are called \emph{positions}. A \emph{reverse} of $w$ is  $w^{-1}:=b_0\cdot b_1\dots b_n$ where $b_0=a_n$, $b_1=a_{n-1}$, \dots, and $b_n=a_0$.

 We  use notation 
	$w_i$ for the prefix of  $w$ ending in position $i\in \nat$.  The \emph{length}  of $w$ is $|w|:= n+1$. Let $\last (w):= |w| - 1$.  We write  $w_{\last}$ to denote $w_{\last(w)}$. Clearly, $w_{\last} =w$.  The empty word, i.e., such that $|w| = 0$, is denoted $\mathbf{e}$. 
The \emph{positions} in word $w$ are denoted as  $i,j\in \nat$. We use $w(i)$ for the $i$-th letter in word $w$. A word  such that  $w(0)=a$, where $a\in \Sigma$, is denoted $w(a)$. 
We use $w(i,j)$ to denote the \emph{subword} obtained from $w$ starting from position $i$ and terminating in position $j$, where $0\leq i \leq j \leq \last(w)$.
Later in the paper, $\Sigma$ will be $\Un$, the set of all $\tau$-structures over the same domain, and, in a Path semantics, words will be called \emph{paths}.

 A \emph{tree} over $\Sigma$ is a (finite or infinite) nonempty set $\Tr\subseteq \Sigma^*$ such that for all $x\cdot c\in \Tr$, with $x\in \Sigma^*$ and $c\in \Sigma$, we have $x\in \Tr$.
 The elements of $\Tr$ are called \emph{nodes}, and the empty word $\mathbf{e}$ is the \emph{root} of $\Tr$. For every $x\in \Tr$, the nodes $x\cdot c\in \Tr$ where $c\in \Sigma$ are the \emph{children} of $x$.  A node with no children is a \emph{leaf}. We refer to the length $|x|$ of $x$ as its \emph{level} in the tree. A \emph{branch} $b$ of a tree $\Tr$ is a set $b\subseteq \Tr$ such that $\mathbf{e} \in b$ and for every $x\in b$, either $x$  is a leaf, or there exists a unique $c\in \Sigma$ such that $x\cdot c \in b$. 
An $|\Sigma|$-ary tree is \emph{full} if each node is either a leaf or has exactly $|\Sigma|$ child nodes. It is \emph{Full} if it has no leaves. In a Full tree, every branch is infinite.

\vspace{-1ex}

	\newcounter{inlineequation}
\setcounter{inlineequation}{0}
\renewcommand{\theinlineequation}{(\Roman{inlineequation})}

\newcommand{\inlineeq}[1]{\refstepcounter{inlineequation}\theinlineequation\ \(#1\)}

\let\sse=\subseteq
\let\vf=\varphi
\def\vc#1#2{#1 _1\zd #1 _{#2}}
\def\zd{,\ldots,}

\colorlet{shadecolor}{gray!30}

\tikzset{My Style/.style={black, draw=shadecolor,fill=shadecolor, minimum size=0.1cm}}

	\begin{figure}[H]
		
\begin{tikzpicture}[
grow= right,
level distance=1.2cm,
level 1/.style={sibling distance=1.0cm},
level 2/.style={sibling distance= 0.6cm},
level 3/.style={sibling distance= 0.6cm},
level 4/.style={sibling distance= 0.58cm},
level 5/.style={sibling distance= 0.38cm}]

\tikzset{level 1 onwards/.style={level distance=0.2cm}}
\tikzset{level 2 onwards/.style={level distance=0.2cm}}
\tikzset{level 3 onwards/.style={level distance=0.4cm}}

\node (Root)(z) {$\emp$}
child {
	node (a)[left=0.2cm] {$\strB$} 
	child { node (b)[left=3mm] {} edge from parent[dashed] 
		node[right,draw=none]{} 
	}
	child { node (c)[left=3mm] {} edge from parent[dashed] 
		node[right,draw=none]{} }
	edge from parent node[left,draw=none] {}
}
child {
	node (d) [left=0.2cm] 
	{$\strA$} 
	child { node (e) {$\strA\cdot\strB$}
		child { node(f) [right=0.2mm]
			 {$\strA\cdot\strB\cdot\strB$}
			child { node(g)[right=1mm] {$\strA\cdot\strB\cdot\strB\cdot\strB$} 
				child { node(kk)[right=2mm] {} edge from parent[dashed]
					node[right,draw=none]{}}
				child { node(ll)[right=2mm] {} edge from parent[dashed]
					node[left,draw=none]{}}
				edge from parent node[right,draw=none]{}}
			child { node(h)[right=1mm][My Style] {$\strA\cdot\strB\cdot\strB\cdot\strA$}
						child { node(k)[right=2mm] {} edge from parent[dashed]
				node[right,draw=none]{}}
			child { node(l)[right=2mm] {} edge from parent[dashed]
				node[left,draw=none]{}}
				edge from parent node[left,draw=none]{} }
			edge from parent
			node[left,draw=none]{}}
		child { node(x)[right=0.2mm] {$\strA\cdot \strB\cdot\strA$} 
			child { node(m) {} edge from parent[dashed]
				node[right,draw=none]{}}
			child { node(q) {} edge from parent[dashed]
				node[left,draw=none]{}}
			edge from parent node[right,draw=none]{}}
	}
	child { node (o){$\strA\cdot \strA$}
		child { node(p)[left=1mm] {} edge from parent[dashed]
			node[right,draw=none]{}}
		child { node (q)[left=1mm] {} edge from parent[dashed] 
			node[right,draw=none]{} }
	}
};

\tikzset{every loop/.style={min distance=10mm,in=160,out=120,looseness=10}}

\tikzset{every loop/.style={min distance=10mm,in=160,out=120,looseness=10}}

\draw[ultra thick,gray] (z)--(d);
\draw[ultra thick,gray] (d)--(e);
\draw[ultra thick,gray] (e)--(f);
\draw[ultra thick,gray] (f)--(h);

\end{tikzpicture}
\caption{A Full tree over $\Un$. The nodes in this tree intuitively correspond to traces of  programs (or proofs). A concrete unconstrained dynamic Choice function $\ch: \Tr \to \Tr$ maps, for example,  $\strA\cdot\strB \mapsto \strA\cdot\strB\cdot\strB$.  
		For each such function $\ch$, there is a unique branch $b$ obtained as follows: $
	b(\ch)  := \{ \emp, \ \ \emp \cdot \ch(\emp),\  \ \emp \cdot \ch(\emp) \cdot \ch(\ch(\emp)),\ \dots  \}$. An example of a branch is traced by the thicker lines.  Each branch $b(\ch)$  determines an equivalence class $[\ch]$ of Choice functions, where   $\ch \cong g  \ \Leftrightarrow \ b(\ch) = b(g)$.  The shaded node is a path, i.e., a finite word $\ppath \subseteq \Un^+$. It is   an  \emph{$\ch$-path}  because  $\ppath \in b(\ch)$.}
\label{fig:tree}
\end{figure}
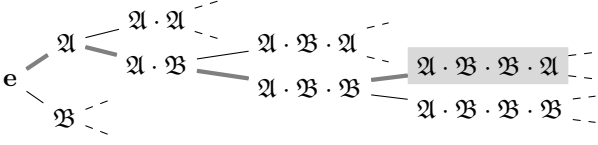

Recall that $\Un$ is the set of all $\tau$-structures over the same domain. Let $\Tr$ be a Full tree over $\Un$. 
We define an \emph{unconstrained} Dynamic Choice function $\ch: \Tr \to \Tr$, with the goal of constraining it with an algebraic term later. 
Function $\ch$, given a node $x$, returns one of its children $x\cdot c$, where $x\cdot c \in \Tr$. The set of all choice functions over $\Un$ is denoted $\CH(\Un)$.
Observe that  for every unconstrained Choice function $\ch$, there is a unique branch $b$ obtained as follows: $
b(\ch)  := \{ \emp, \ \ \emp \cdot \ch(\emp),\  \ \emp \cdot \ch(\emp) \cdot \ch(\ch(\emp)),\ \dots  \}$. This is because, at each level of the Full tree, $\ch$ makes a unique choice.
However, Choice functions are defined for every node in $\Tr$, not just those reachable through a repeated application of $\ch$ from the root $\mathbf{e}$ of $\Tr$. For each branch, there are infinitely many unconstrained Choice functions that agree on that branch, but act differently  on nodes not in that branch. 
We consider such functions equivalent if and only if the corresponding branches are  the same: $\ch \cong g  \ \Leftrightarrow \ b(\ch) = b(g)$. Each branch $b(\ch)$  determines an equivalence class  of Choice functions.

Let a  \emph{path} be  a finite word $\ppath \subseteq \Un^+$. 
\begin{definition} \label{def:h-path} We call a path $\ppath$ an  \emph{$\ch$-path}  if $\ppath \in b(\ch)$.  
	\end{definition}

\vspace{2ex}

\section{An Algebra of Binary Relations: Syntax and Atomic Transitions}
\label{sec:Algebra-syntax}

In this section, first, we introduce the syntax of an algebra of binary relations. Second, we semantically formalize the important inertia preservation principle for \emph{atomic} transductions. This principle has first appeared under the name of the commonsense law of inertia \cite{McCHay69}. Third, we introduce a specific kind of such transductions that we use in this paper. The trasductions are a singleton-set modification of monadic conjunctive queries. Such queries have multiple outcomes because only one (arbitrary) domain element is included in each output. Finally, we give an intuition about the underlying machine model.

\subsection{Algebraic Expressions}

We give an algebraic syntax of our logic first. We call this syntax \emph{one-sorted} to contrast it with the two-sorted one in the form of a  modal Dynamic logic, presented in Section \ref{sec:dynamic-logic}, where we have sorts for state formulae and processes.
 Algebraic expressions are given by the following grammar:
\begin{equation}
\label{eq:one-sorted}
 \begin{array}{c}	
 \ \ \ 	\alpha  :: =   
	\id  \mid \varepsilon M (\bZ)  \mid    \rneg\alpha  \mid    \lneg \alpha \mid  
	\alpha\comp \alpha  \mid \\
	\alpha  \sqcup \alpha \mid    \sigma^r_{X=Y} (\alpha) \mid    \alpha ^\iter  \mid   \alpha^\converse .
\end{array}
\end{equation}
Here, $M$ is any module name in $\Sch{}$,  $\bZ$ is a tuple of variables; and $X,Y$ are variables. For \emph{atomic module expressions}, i.e., expressions of the form $M(\bZ)$, the length of $\bZ$ must equal $\ar(M)$. In practice, we will often write $M (\underline{\bX};\bY) $ for atomic module expressions, where $\bX$ is a tuple of variables of length $\iar(M)$ and $\bY$ is a tuple of variables of length $\ar(M)- \iar(M)$.
These variables are relational.
In addition, we have: Identity (Diagonal) ($\id$), atomic module symbols preceded by  Choice  operator ($\varepsilon M(\bZ)$), where $\varepsilon$  denotes a free Choice function variable, Sequential Composition ($ \comp$), Forward-facing ($\rneg  $) and Backwards-facing ($\lneg$ ) Unary Negations, also called Anti-Domain and Anti-Image, respectively, Preferential Union ($\sqcup$), which is a restriction of Union, Right Selection $\sigma^r_{X=Y}$ which amounts to equality check for the two variables on the right of a binary relation specified by $\alpha$,
Maximum Iterate ($^\iter$), which is a restriction of  transitive closure, and Converse $\alpha^\converse$.\footnote{In (\ref{eq:one-sorted}), Right Selection $\sigma^r_{X=Y}(\alpha)$ can be replaced with \emph{equality test}, i.e., $\sigma^r_{X=Y}(\id)$, which could be written as just $X=Y$.}

Our algebra  has a  \textbf{two-level syntax}: in addition to the algebraic expressions (\ref{eq:one-sorted}), atomic transductions provide a transition system context for the algebra. Each atomic transduction is specified as presented in Sections \ref{sec:inertia}   and \ref{sec:SM-PP}.  The detailed logical formalization is needed to make our proofs fully formal. However, at a first reading, we recommend skipping these two sections to go directly to Section \ref{sec:machine-model}.

\subsection{Semantic Context:  Atomic Transitions}
\label{sec:inertia}
Next, we formalize atomic transductions semantically, while taking into account inertia preservation.

	A \emph{semantic context} $\inst{}:= \inst(s,\Domain,S(\cdot))$ is a function that maps each atomic symbol $M$ in $\Sch{}$ to a binary relation (here denoted $\sem{\cdot}$) on the set $\Un$ of all $\tau$-structures over the domain $\Domain$. This, a semantic context $\inst{}$ gives a transition relation on $\Un$.
	The function $\inst{}$ is parameterized with 
	\begin{compactitem}
	\item a variable assignment  (an instantiation function)  $s : \vuni \to \tau$, where $\tau$ is a relational vocabulary;
	\item a base set $\Domain$, which is, in this paper, the domain of the input structure;
	\item a \emph{static} interpretation $S(M)$ of module names $M$ in $\Sch{}$.
For each atomic module name $M$ in $\Sch{}$ and tuples of variables $\bZ$ and $\bX\subseteq \bZ$, with  $|\bZ| = \ar{(M)}$ and  $|\bX| = \iar{(M)}$,  static interpretation $S(M)$ is a
function $S$ that returns  
  a set of  $s(\bZ)$-structures ($s(\bZ) \subseteq \tau$) with domain $\Domain$, where an \emph{interpretation of $s(\bX)$  is expanded to the entire vocabulary  $s(\bZ)$ to satisfy a specification in a logic $\cL_T$ to be specified in Section \ref{sec:SM-PP}}.\footnote{We abuse the notations slightly by writing $s(\bZ)$ to denote the set $\{s(Z_1), s(Z_2)\dots\}$, where $Z_1,Z_2,\dots$ are variables in $\bZ$.} 
\end{compactitem}

To finish defining semantic context $\inst{}:= \inst(s,\Domain,S(\cdot))$, we reinterpret each static interpretation dynamically, that is, as a \textbf{transduction} represented by a \textbf{dynamic relation (DR)}, i.e., a binary relation on structures. 
  For \emph{atomic} DRs, we require that, relations, that are not explicitly modified by the expansion from $s(\bX)$  to  $s(\bZ)$, remain the same (are transferred by inertia). This includes interpretations of unconstrained variables that are interpreted arbitrarily. 
Taking into account this intuition, we present atomic DRs formally, as follows.
Static interpretation $S(M)$, for each atomic module name $M$,  gives rise to a binary relation $\inst(M)$ (here denoted $\sem{\cdot}$) on $\Un$ defined as follows:
$$
\begin{array}{l}
\hspace{-5mm} \sem{M(\underline{\bX};\bY)}    : =  \big\{ (  \strA ,  \strB) \in  \Un \! \times\!  \Un  \mid 

\mbox{  exists }  
\strC \in S(M) \mbox{ and  } 
\end{array}
$$ 	
\vspace{-4mm}
\begin{equation}\label{eq:semantics-of-atomic-modules}
	\begin{array}{c}	
		\strC|_{ s(\bX) } =  \strA|_{  s(\bX) }, \\ 
		\strC|_{ s(\bY) } =  \strB|_{  s(\bY)  }  
		\\  \mbox{ and }  \  \strA|_{\tau\setminus s(\bY) } =  \strB|_{\tau\setminus  s(\bY) }  
		\big\}.
	\end{array}
\end{equation}

\noindent  Intuitively, the dynamic semantics of an expression $M(\underline{\bX};\bY)$ represents a transition from $\strA$ to $\strB$.  The inputs of the module are ``read'' in $\strA$ and the outputs are updated in $\strB$. Relations unaffected by the update are copied from  $\strA$ to $\strB$, as formalized by the last line of (\ref{eq:semantics-of-atomic-modules}).

The  interpretation of atomic elements as binary relations on structures means that the algebra (\ref{eq:one-sorted}) is an \emph{algebra of transductions} or structure-rewriting operations.

\subsection{Singleton-Set Monadic  PP-definable (SM-PP) Atomic Transductions}

Next, we specialize the logic $\cL_T$ for axiomatizing expansions from $s(\bX)$ to  $s(\bZ)$.
This logic gives us another parameter to control the expressive power, in addition to the algebra (\ref{eq:one-sorted}).
Our goal is to  restrict atomic transductions so that the only relations that change from structure to structure are unary and contain one domain element at a time, similarly to  registers in a Register Machine \cite{ShepherdsonSturgis:1961}. We take a \emph{monadic} primitive positive relation (i.e., a relation definable by a conjunctive query, see Technical Preliminaries), and output only a \emph{single element contained in the relation instead of the whole relation}.

\begin{definition}\label{def:SM_PP}\rm
	\emph{Singleton-Monadic  Primitive Positive ({\rm SM-PP}) relation} is a \textbf{singleton set} relation $R$ implicitly definable by:
	\begin{equation}\label{eq:SM-PP}
	\begin{array}{c}
	\forall x \ \big(  
	R(x) \to  \underbrace{\exists z_1 \dots \exists z_m\ (B_1 \land \cdots \land B_m)}_{\sf Conjunctive\   Query}\big) \big)\\
	\land \ \ 	\forall x	\forall y \ \big(  
	R(x) \land 	R(y)  \to x=y \big),
	\end{array}
	\end{equation}
	where each $B_i$ is a relational variable from $\vuni$ with object variables from 
	$x, z_1, \dots , z_m$ such that it is either (1) unary or (2) the interpretation of the vocabulary symbol $s(B_i)$ comes from the input structure $\strA$ (i.e., it is EDB in the database terminology). 
We use a rule-based notation for (\ref{eq:SM-PP}):
	\begin{equation}\label{eq:SM-PP-rule}
	\ \ \ \ \ \ \ \ \ \ \ \ \ R(x) \rul B_1,  \dots, B_m.
	\end{equation}
\end{definition}
\noindent Notation $\rul$, unlike Datalog's $\leftarrow$, is used to emphasize that only one element is put into the relation in the head.

\begin{definition}	
	A \emph{{\rm SM-PP} atomic transduction (or atomic module)} is a set of rules of the form (\ref{eq:SM-PP-rule}). 
	\end{definition}
\noindent Thus, SM-PP transductions have:

	\vspace{1ex}

\hspace{4ex}	\textbf{Outputs:} SM-PP definable relations,
	
\hspace{4ex}	\textbf{Inputs:} monadic or given by the input structure $\strA$.

\vspace{2ex}

\begin{example}{\rm A non-example of SM-PP atomic transduction is a computation of path of length two from  Example \ref{ex:Path2}. The relation  on the output is neither monadic nor singleton-set.
} 
	\end{example}

\begin{example}\label{ex:Reach2-rule}
	{\rm  Consider the monadic PP-definable by (\ref{eq:Reach2}) relation $Z$ from Example \ref{ex:Reach2}, and modify that expression as in Definition \ref{def:SM_PP}. We obtain, in the rule-based syntax (\ref{eq:SM-PP-rule}), an SM-PP atomic transduction: 	
	\begin{equation}\label{eq:Reach2-rule}
	\defin{Z(x_2) \rul X(x_1), Y(x_1, z) , Y(z, x_2))}.
\end{equation}
} 
\end{example}

\begin{example}\label{ex:Reach2-cont}\rm Let $M_{R2}$ in  module vocabulary $\Sch{}$  be the name of the transduction axiomatized by (\ref{eq:Reach2-rule}) from Example \ref{ex:Reach2-rule}. Let  $\tau$ contains predicate symbols $E$, $S$ and $R2$, among other symbols, and the  instantiation function $s$ maps relational variables $X$, $Y$ and $Z$ to $S$, $E$ and $R2$, respectively.
	Then static interpretation $S(M_{R2})$ is a set of all $s(X,Y,Z)$-structures over some domain $\Domain$   that are the models of the expression (\ref{eq:Reach2-rule}).

	Now, suppose, according to the module \schema  $\Sch{}$, we have: $M_{R2}(\underline{X},\underline{Y},Z)$, where the variables $X$ and $Y$ in the input positions are underlined.  
	Then by (\ref{eq:semantics-of-atomic-modules}), we have a binary relation on $\tau$-structures, with $s(X,Y,Z)\subseteq \tau$.   Edges $E = s(Y)$ and the start node $S =s(X)$ are ``read'' in  a $\tau$-structure $\strA$ on the left, and the set of reachable in two steps  nodes $R2= s(Z)$ is updated in $\strB$ on the right, see Figure \ref{fig:MS-state-change}.

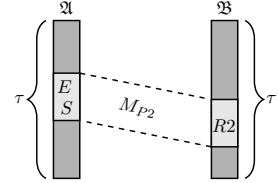
	\begin{minipage}{.22\textwidth}
		\begin{figure}[H]
			\hspace{3ex}
			\begin{tikzpicture}[thick,transform canvas={scale=0.7}]
			\node (R1A) {};
			\path (R1A)+(0.5,-3) node (R1B) {};
			\draw [fill=black!30] (R1A) rectangle (R1B);
			
			\path (R1A)+(0,-1) node (R1sA) {};
			\path (R1sA)+(0.5,-.9) node (R1sB) {};
			\draw [fill=black!10] (R1sA) rectangle (R1sB);
			
			\path (R1A)+(3,0) node (R2A) {};
			\path (R1B)+(3,0) node (R2B) {};
			\draw [fill=black!30] (R2A) rectangle (R2B);
			
			\path (R1sA)+(3,-0.5) node (R2eA) {};
			\path (R1sB)+(3,-0.5) node (R2eB) {};
			\draw [fill=black!10] (R2eA) rectangle (R2eB);
			
			\path (R1sA)+(0.5,0) edge[-,dashed] node[sloped,anchor=center,below,yshift=-0.15cm]{$M_{P2}$} (R2eA);
			\path (R2eA)+(0,-0.9) edge[-,dashed] (R1sB);
			
			\path (R1sA)+(0,-0.25) node[right,xshift=-0.06cm] {$E$};
			\path (R1sA)+(0,-0.25) node[right,yshift=-0.4cm] {$S$};
			\path (R2eB)+(0,0.25) node[left,xshift=0.09cm,yshift=0.15cm] {$R2$};
			
			\draw [decorate,decoration={brace,amplitude=10pt},xshift=-4pt,yshift=0pt] (0,-3) -- (0,0) node [black,midway,xshift=-0.5cm] {$\tau$};
			\draw [decorate,decoration={brace,amplitude=10pt},xshift=4pt,yshift=0pt] (3.5,0) -- (3.5,-3) node [black,midway,xshift=0.5cm] {$\tau$};
			
			\path (0.25,0) node [above] {$\strA$};
			\path (3.25,0) node [above] {$\strB$};
			\end{tikzpicture}
			\vspace{20mm}	
			\caption{An illustration of  (\ref{eq:semantics-of-atomic-modules}).	}\label{fig:MS-state-change}
		\end{figure}
	\end{minipage} \hfill
	\begin{minipage}{0.22\textwidth}
			By the last line of (\ref{eq:semantics-of-atomic-modules}), the interpretations of all symbols except the updated $R2$ are preserved in the transition from $\strA$ to $\strB$. 
	\end{minipage}	
\end{example}

\vspace{2ex}

\begin{example}\label{ex:same-generation-atomic}{\rm The following atomic transduction is SM-PP:
	\begin{equation}\label{ex:rules-simult-reach}
	\begin{array}{l}
	\hspace{-2mm}  
  \defin{  \Reach_A'(x ) \rul \underline{\Reach_A(y)}, \underline{E(x,y)},\\
		\Reach_B'(v ) \rul \underline{\Reach_B(w)}, \underline{E(v,w)} }.
	\end{array}
	\end{equation}
Here, again, we underline the variables in  $\vuni$ that appear in the input positions. We also use more informative names for the relational variables, e.g., $\Reach_A$ instead of, say, $X,Y$ or $Z$, etc. It will be our practice from now on. This module is, in fact, a part of Same Generation example presented in Section \ref{sec:Examples}. 
} 
\end{example}

\subsection{Machine Model}\label{sec:machine-model}

In the context of SM-PP atomic transductions, 
we can think of evaluations of algebraic expressions (\ref{eq:one-sorted}) as  \textbf{computations  of Register machines starting from input $\strA$}. 
Importantly, we are interested in \textbf{isomorphism-invariant} computations of these machines, i.e., those that do not distinguish between isomorphic input  structures, as required by Definition \ref{def:logic} of a proper Logic.
Intuitively, we have: 
\begin{itemize}
	\item monadic ``registers'' -- predicates used during the computation, each containing only one element at a time, e.g.,\\ $\Reach_A(y),\ \ \Reach'_A(x),\ \ \Reach_B(w),\ \ \Reach'_B(v) $;
	
	\item  the ``real'' inputs, e.g., ${\bf E}(x,y)$, are of any arity;
\item atomic transitions correspond to conditional assignments with a non-deterministic outcome;
\item in each atomic step, only the registers of the current state or the input structure are accessible;	
		\item  the Choice operator $\varepsilon$, depending on the history,
	chooses one of the possible outputs, e.g., an interpretation 
	$({\Reach'_A}^{\strA_{i+1}},{\Reach'_B}^{\strA_{i+1}})$ of ``registers'' $\Reach'_A$, $\Reach'_B$ in the next state ${\strA_{i+1}}$ in the Example \ref{ex:same-generation-atomic} (formally, this operator is defined in Section \ref{sec:Algebra-semantics});	
	\item computations are  controlled by  algebraic terms that, intuitively, represent programs.
\end{itemize}
A related notion of a Register Kripke structure will be given in Section \ref{sec:register-structures}.
In Section \ref{sec:programming-constructs}, we shall see that the main programming constructs are definable. Thus, programs for our machines are very close to standard Imperative programs. In Section \ref{sec:Examples}, we give examples of ``programming'' in this model.

\label{sec:SM-PP}

\vspace{2ex}

\section{Semantics of the Algebra}
\label{sec:Algebra-semantics}

In this section, we introduce a semantics for the algebraic operations  and define a binary satisfaction relation for algebraic terms.

\subsection{Global Dynamic Relations}

We define the  semantics of (\ref{eq:one-sorted}) as an algebra of binary relations within a semantic context $\inst{}$ (see (\ref{eq:semantics-of-atomic-modules})), which gives interpretations to atomic transitions.  In addition, the semantics is conditional on a specific instantiation $\ch$ of the free Choice variable $\varepsilon$.
Recall, from Section \ref {sec:inertia}, that $\bbar$ stands for Dynamic Relation.
The semantics of an expression $\alpha$ is given as a $\bbar$ which will be
denoted by $\sema\alpha$.  We, again, adapt
 Gurevich's terminology
\cite{Gurevich-algebra-feasible,Gurevich-challenge}, as we did in Section \ref{sec:Preliminaries} (Technical Preliminaries). We say that every algebraic  expression $\alpha$
denotes a \emph{\gbrv} $\semNoI{\alpha}$: a function that maps a binary interpretation $\inst{}$  of
$\Sch$ and an instantiation $\ch$ of $\varepsilon$ to the \bbar $\alpha(\inst{},\ch) := \sema{\alpha}$.\footnote{We can also view algebraic expressions as  binary queries on labelled chains.}
Later on, we will see a way of turning algebraic expressions into boolean queries. Intuitively,  such a query will ask if there is a successful execution of $\alpha$ in the input structure, or  whether $\alpha$ ``accepts'' that input. This way, algebraic expressions will be used to represent computational problems viewed as isomorphism-closed classes of structures.  Choice functions (or, equivalently, traces of computation, see Figure \ref{fig:tree} for an intuition) will serve as witnesses for the membership in such problems. 

\subsection{Semantics of Function-Preserving Fragment}
We now go back to the semantics of the algebra (\ref{eq:one-sorted}).
First, we define the semantics of a \textbf{Function-Preserving Fragment}, which is (\ref{eq:one-sorted}) without $\e$ and Converse. The name comes from the fact that all operations in this fragment preserve functions. 
That is, if atomic relations are functional, then so are the algebraic expressions in this fragment. By defining the operations this way, we ensure that they do not ``add guessing''. However, atomic modules are allowed to be general relations. Moreover, notice that we do not include intersection that may potentially \emph{reduce} the expressive power of the components due to confluence.

The meaning of the subscript $\ch$ used below will become clear when we define the semantics of the  free Choice variable $\varepsilon$. As we mentioned before, function $\ch$ will be a specific semantic instantiation of this variable.  Choice plays no role in the Function-Preserving Fragment (since the fragment does not have the $\varepsilon$ operator).

Atomic module names $M$ in $\Sch{}$ are interpreted  as in (\ref{eq:semantics-of-atomic-modules}), i.e., as binary relations on structures $
\sema{ M (\underline{\bX};\bY) }     \subseteq  \Un \times  \Un  
$.

\noindent We now extend the binary interpretation $\sema{\cdot}$  to all 
 expressions, starting with the operations of the Function-Preserving fragment. 

\vspace{2ex}

 \noindent{\bf Identity (Diagonal)}
$
\begin{array}{l}
  \sema{\id}    : =   \{ (  \strA ,  \strB) 
  \in  \Un \times  \Un  \mid \strA =  \strB  \} .
\end{array}
$
Operation $\id$ is sometimes called the ``$nil$'' action. It is the empty word  in the formal language theory \cite{HopcroftUllman:1979}.

 \vspace{2ex}

\noindent{\bf Right Selection (Equality)} $
 \sema{ \sigma^r_{ X \eq  Y} (\alpha) }  : = \{ ( \strA ,  \strB ) \in  \Un \times  \Un  \mid
  ( \strA ,  \strB ) \in
 \sema{ \alpha  } \mbox{ and } (s(X))^\strB =  (s(Y))^\strB
 \}.
 $

 \vspace{2ex}

\noindent{\bf Sequential Composition} $
 \sema{  \alpha  \comp  \beta  }  : = \{ (  \strA ,  \strB) \in  \Un \times  \Un  \mid  
 \exists  \strC\  (( \strA ,  \strC ) \in \sema{  \alpha }
 \mbox{ and }   (  \strC ,  \strB ) \in \sema{ \beta  }   )  \}.
 $
 
 \vspace{2ex}

\noindent{\bf Forward Unary Negation (Anti-Domain)}
Note that regular complementation includes all possible transitions 
except $\alpha$. We introduce a stronger negation which 
is essentially unary (expressed as binary with two identical elements in the pair), and excludes states where $\alpha$ originates.\footnote{Ideally, these operations should be denoted with a notation closer to the typical negation sign, i.e., more ``square'',  however we could not find a suitable pair of Latex symbols. Any suggestions for a better solution are welcome.} 
It says ``there is no outgoing $\alpha$-transition''.  The semantics is:

$
\sema{  \rneg \alpha  }  : = \{ (  \strB ,  \strB)  \in  \Un \times  \Un  \mid   \ \forall  \strB'\ \  (  \strB ,  \strB') \not \in \sema{ \alpha } \}.
$

\vspace{2ex}

\noindent{\bf Backwards Unary Negation (Anti-Image)}
A similar operation for the opposite direction says ``there is no incoming $\alpha$-transition''. Thus,
$
\sema{  \lneg\alpha  } : = \{ (  \strB ,  \strB)   \in  \Un \times  \Un  \mid   \ \forall  \strB'\ \  (  \strB' ,  \strB) \not \in \sema{ \alpha } \}.
$
Each of the unary negations is a restriction of the regular negation (complementation). 
 Unlike regular negation, these unary operations preserve the property of being functional.
The connectives $\rneg$\ and $\lneg$\  have the properties of the Intuitionistic negation: $\rneg\rneg \rneg \alpha = \rneg \alpha$, $\rneg\rneg \alpha \neq \alpha$, and the same holds for the other direction.

\vspace{2ex}

\noindent{\bf Maximum Iterate}
Maximum Iterate is a determinization of  the Kleene star $\alpha^*$  (reflexive transitive closure). It 
outputs only the longest, in terms of the number of $\alpha$-steps, transition out of all possible transitions produced 
by the Kleene star. 
We define it as: 
$\alpha ^\iter  =  \bigcup_{0\leq n < \omega } (\alpha_n \comp \rneg \alpha ) $,
 where
$\alpha_0 :=  \id$,  and 
$\alpha_{n+1}:=  \alpha \comp \alpha_n$.\footnote{For readers familiar with the   least fixed point operator $\mu Z.\phi $, we mention  that Maximum Iterate is definable as:
$
\begin{array}{l}
\alpha ^{\iter}\ :=\ \mu Z. ( \rneg \alpha \cup \alpha \comp Z).
\end{array}$ From this definition, it is clear that infinite ``looping'' is excluded by the base case.}

\vspace{2ex}

\noindent{\bf Preferential union}
$$
\begin{array}{cc}
\sema{ \alpha \sqcup \beta  } : =  \Big\{ (  \strA ,  \strB)  \in  \Un \times  \Un     & \\ & \hspace{-3.8cm} \left| 
\begin{array}{cl}
(  \strA ,  \strB) \in \sema{ \alpha  }     &\mbox{if } (\strA,\strA) \in \sema{ \Dom(\alpha)   },  \\
(  \strA ,  \strB) \in \sema{ \beta  }   & \mbox{if }    (\strA,\strA) \not\in \sema{ \Dom(\alpha)   }  \\   {\ \ \ \ \ \ \ \ \ \ \ \ \ \ \ } &\mbox{and  }  (\strA,\strA) \in \sema{ \Dom(\beta)   }
\end{array} \! \! \! \right\},
\end{array}
$$ 
where $\Dom(\alpha):= \rneg \rneg \alpha$ is a projection onto the left element of the binary relation. Thus, we perform $\alpha$ if it is defined, otherwise we perform $\beta$. This operation is a determinization of Union.\footnote{The connection with Union (a construct not included in the language) is as follows: $
	\alpha \sqcup \beta \ = \ \alpha \cup ( \rneg \alpha  \comp \beta).
	$ }

 \subsection{Semantics of Epsilon and Converse}
Having defined the semantics of the Function-Preserving Fragment, we now define the semantics of the full algebra (\ref{eq:one-sorted}). For  that,  we introduce the semantics of the free Choice variable $\varepsilon$ and Converse.

\vspace{1ex}

 \noindent{\bf  Dynamic Version of Hilbert's Choice Operator Epsilon }
Let $\Tr$ be a Full tree over $\Un$, and $\ch$ be an unconstrained specific Choice function  $\ch: \Tr \to \Tr$.  As we already mentioned,  $\ch$ is a parameter of the semantics that affects  the result of the evaluation of an  expression. It serves as an instantiation of the free function variable $\varepsilon$. Thus, we have:
\begin{equation} \label{eq:epsilon-semantics}
\begin{array}{l}
\hspace{-3mm}
\sema{ \varepsilon M(\underline{\bX};\bY)  }    : =  \\ \hspace{-3mm} \big\{   
\big(  w(\last),  w'(last)\big) \in \sem{  M(\underline{\bX};\bY) }  \  \mid \  w'=\ch(w),\\
\hspace{28ex} \mbox { and } w, w' \in b(\ch)  \big\}.
\end{array}
\end{equation}		
\noindent By this definition, the Epsilon operator, for each $\tau$-structure $\strA = w(\last)$ where $M$ is defined, arbitrarily selects \emph{precisely one} structure $\strB = w'(\last)$ out of all possible structures that are $M$-successors of $\strA$, as prescribed by  $\ch$  on the branch $b(\ch)$. Notice that $\ch$-choices outside $b(\ch)$ are not included. Branch $b(\ch)$ will also be used in the connection with a Path semantics defined in Section \ref{sec:Path-Semantics}.
For example, in Figure \ref{fig:tree}, $\ch: w \mapsto w'$, where $w=\strA\cdot\strB\cdot\strB$, $w'=\strA\cdot\strB\cdot\strB\cdot\strA$,  $w(last)=\strB$ and $w'(last)=\strA$. If  $(\strB,\strA)\in \sem{  M(\underline{\bX};\bY) }$, then $(\strB,\strA)\in \sema{ \e M(\underline{\bX};\bY) }$, since $w,w' \in b(\ch)$.

\vspace{2ex}

\noindent \textbf{The semantics of Converse} For the atomic case of \  $^\converse$:
$$
\hspace{-3.5mm}\begin{array}{l}
\sema{ \varepsilon M(\underline{\bX};\bY)^\converse  }    : =  \big\{ (  \strB ,  \strA)\!  \in \! \UnT \! \times\!  \UnT  \mid   
(  \strA ,  \strB)\! \in\!  \sema{ \varepsilon M (\underline{\bX};\bY)} \big\}.
\end{array}
$$
We extend this semantics  to all operators of (\ref{eq:one-sorted}) inductively:
\begin{equation}\label{eq:converse}
\begin{array}{ll}
(	\alpha\comp \beta )^\converse  \ : = \ 	\beta^\converse\comp \alpha^\converse,&
(\rneg \alpha)^\converse  \ : = \ \rneg \alpha,\\ 
(\lneg \alpha)^\converse  \ : = \ \lneg \alpha,&

(\alpha ^\iter)^\converse   \ : = \ \alpha^{\converse \ \iter},\\
(\alpha ^\converse)^\converse  \ : = \ \alpha,&
(\sigma^r_{\Theta} (\alpha))^\converse   \ := \  \sigma^r_{\Theta} (\alpha^\converse).
\end{array}
\end{equation}	
Observe that, since the semantics is parameterized with a  Choice function $\ch$, Converse simply reverses the relation  along the same linear trace.

\subsection{Binary Satisfaction Relation}\label{sec:binary-satisfaction-relation}
Given a  well-formed algebraic term $ \alpha$, we say that pair of structures 
$(\strA, \strB)$, \emph{satisfies} $ \alpha$ under 
 semantic context $\inst{}$ and choice function $\ch$
if and only if $(\strA, \strB ) \in  \sema{ \alpha   }$:
\begin{equation} \label{eq:binary-sat}
(\strA, \strB )\models_\inst{}  \alpha[\ch/\e] \ \ \Leftrightarrow\ \  (\strA, \strB ) \in  \sema{ \alpha   }.
\end{equation}
Here,  the notation $[\ch/\e]$ means a semantic instantiation of the free function variable $\e$ by a specific Choice function $\ch$.

\vspace{1ex}

\paragraph{(Non-)Deterministic Transitions}We sometimes use the term \emph{deterministic transitions}  for functional binary relations. A transition is \emph{non-deterministic} if the corresponding  relation is not functional.

\vspace{2ex}

\section{Two-Way Deterministic Linear Dynamic  Logic (DetLDL$_f$)}\label{sec:Dynamic-Logic}

In this section, starting from the algebra (\ref{eq:one-sorted}), we define an equivalent syntax in the form of a  Dynamic logic. While the two formalizations are equivalent, in many ways, it is much easier to work with the logic than with the algebra, as will already become clear later in this section. In particular, we shall be able to formalize computational problems, and later, in Section \ref{sec:Path-Semantics}, introduce a language-theoretic perspective via an alternative Path semantics for the logic.
In this section, after defining the Dynamic logic, we justify why composition is sufficient to express conjunction. Then, we define a unary satisfaction relation and a query we use over and over again in this paper. The query, intuitively,  checks if there is a successful execution of a process starting from an input structure. We demonstrate that, using the function-preserving process operations of the Dynamic logic, one can define all major constructs of Imperative programming.  Finally, we formalize the notion of a computational problems specified by an algebraic term and discuss two forms of certificates for membership in such problems. 

\subsection{Dynamic Logic}\label{sec:dynamic-logic}
We now provide an alternative (and equivalent) two-sorted version of the syntax  (\ref{eq:one-sorted})  in the form of a modal temporal (dynamic) logic that we call DetLDL$_f$, by analogy with LDL$_f$ \cite{DeGiacomo-Vardi:IJCAI:013}. The syntax is given by the grammar: 
\begin{equation}
\label{eq:LLIF-two-sorted}
\begin{array}{c}
\alpha \ :: =  \  
\id   \mid  \varepsilon M_a (\bZ)  \mid   \\ \rneg \alpha  \mid    \lneg \alpha  \mid 
\alpha\comp \alpha     \mid    \alpha  \sqcup \alpha \mid    \sigma^r_{X=Y} (\alpha) \mid    \alpha ^\iter  \mid \alpha^\converse \mid  \phi?   \\

\phi :: =  \mT \mid  M_p (\bZ)   \mid  \neg  \phi  \mid | \alpha \rangle \ \phi  \   | \ \langle \alpha | \  \phi.
\end{array}
\end{equation}
The subscripts $a$ and $p$ stand for ``actions'' and ``propositions'', respectively. Intuitively, proposition modules $M_p$ make only self-loop transitions.

The expressions in the first two lines of (\ref{eq:LLIF-two-sorted}) are typically called \emph{process} terms, and those in the third line \emph{state} formulae. 
	The idea is that state formulae $\phi$ in the third line of (\ref{eq:LLIF-two-sorted}) are ``unary'' in the same sense as $P(x)$ is a unary notation for $P(x,x)$. Semantically, unary formulae are subsets of the identity relation on $\Un$.
State formulae $|\alpha \rangle \ \phi$ and $\langle \alpha | \  \phi$  are right and left-facing existential modalities. 
The  state formulae in the second line of (\ref{eq:LLIF-two-sorted}) are 
 shorthands that use the operations in the first line: $ \mT : = \id$, $ \neg  \phi : = \rneg \phi$ (or $  \lneg \phi  $), \ $  | \alpha \rangle \ \phi := \rneg \rneg \alpha \comp \phi$, \ \ $  \langle \alpha | \  \phi := \lneg \lneg \phi\comp \alpha$  and  $ \phi ? :=  \rneg \rneg  \phi$.
Recall that we have defined $ \Dom (\alpha) := \ \rneg \rneg \!\! \alpha $, which, intuitively, denotes the domain of $\alpha$. With this notation, we have that $  | \alpha \rangle \ \phi \ = \ \Dom (\alpha \comp \phi)$.

While we do not use a box modality in this paper, we introduce it for completeness of the exposition. We  define $|\alpha ] \ \phi \  : =\  \rneg (\alpha\comp  \rneg\phi)$. Taking into account Weak Idempotence of the Unary Negation, $\rneg \rneg\rneg \!\! \phi = \rneg  \!\! \phi$, we obtain that the two modalities are duals of each other,  $|\alpha \rangle \ \phi \  =\  \rneg |\alpha ] \rneg \phi $.
Note that, since $\ch$ is a parameter of the semantics, the duality is specific to that $\ch$ (or, equivalently, to an $\ch$-path, as will be clear from Path semantics in part (a) of Section \ref{sec:path-semantics}). 
From the that section, \ref{sec:path-semantics}, the universal and existential meaning of the two dual modalities will also be more transparent.

We often write $M_p(\underline{\bX};\ )$ as a particular case of $M_p( \bZ)$, where $\bX \subseteq \bZ$, to indicate that only input variables are checked, and no transition to another state happens. Semantically, we have, for all $\ch$: 
$$
\begin{array}{c}
\sema{  M_p(\underline{\bX}; )  }     =  \sema{ \varepsilon M_p(\bZ)  }     \ := \\
\hspace{-2ex}\big\{   
\big(  w(\last),  w(\last))\big) \in \sem{  M_p(\underline{\bX};\bY) }    \mid    
 \sem{  M_p(\underline{\bX};\bY) } \subseteq \id  \big\}.
\end{array}
$$
Thus, Choice functions do not affect the semantics of modules-propositions. Such modules act as test actions:  $ M_p(\underline{\bX}; ) ? =  \rneg \rneg \! M_p(\underline{\bX}; ) =  M_p(\underline{\bX}; )$.

\subsection{Composition as a Form of Conjunction} Intersection is not included among the operations of the algebra because it ``wastes'' computation already performed by the intersected components.
One may think that this operation is needed to express conjunction of state formulae. This is not the case. 
It has been known for a long time that relational composition, also called relative or dynamic product, is an analog of logical or static conjunction. According to Pratt \cite{Pratt:calc-bin-rel}, ``This view of composition/concatenation as a form of conjunction predates even Peirce and would appear to be due to De Morgan in 1860 \cite{DeMorgan:1860}''. In some Resource logics, Composition is called Multiplicative Conjunction.
It is not hard to check that,   in the unary (state formulae) case,  `$\comp$' and `$\sqcup$'  have the properties of `$\land$' and `$\lor$', respectively. We could have added the case of conjunction to the state formulae of (\ref{eq:LLIF-two-sorted}), but prefer not to, to avoid an extra case in the proofs. Our examples are $\land$-free, but one can add it when needed.

\subsection{Unary Satisfaction Relation and  Main Query}\label{sec:unary-satisfaction}
While, technically, the two-sorted syntax does not add any expressiveness, it provides a considerable usability advantage over the unary one. Indeed, we now have a Dynamic logic where terms, that informally represent programs, can occur inside the possibility and the necessity modalities. The logic allows one to conveniently specify properties of such programs,  for example, to define the set of all \emph{successful} executions \emph{for a specific input}. This leads us to the notion of a unary satisfaction relation.

\vspace{1ex}

\subsubsection{Unary Satisfaction Relation}For state formulae $\phi$, we define: 
\begin{equation} \label{eq:unary-sat}
\strA \models_\inst{}  \phi[\e] \ \ \Leftrightarrow\ \   (\strA, \strA )\models_\inst{}   \phi[\e], 
\end{equation} 
to say that that formula $\phi$ is \emph{true} in structure $\strA$ under some semantic  instantiation  $h\in\CH(\Un)$ of the free function variable $\e$.
 Thus, unary satisfaction relation is a particular case of the binary one, which is defined in Section \ref{sec:binary-satisfaction-relation} by (\ref{eq:binary-sat}).
Here, we attached the (optional) notation $[\e]$ to $\phi$ to emphasize the presence of the free function variable in the formula. We shall continue doing it for clarity, when needed.

According to (\ref{eq:unary-sat}),
state formulae $\phi$ in logic (\ref{eq:LLIF-two-sorted}) are interpreted, for each semantic  instantiation  $h\in\CH(\Un)$ of the free function variable $\e$, by subsets of $\Un$
 where they are satisfied. 
On the other hand,  formula $\phi[\e]$ defines all possible Choice functions $\ch$  that witness $\strA \models_\inst{}   \phi[\e] $, for each input structure $\strA$.

\subsubsection{Main Query} We often write
\begin{equation}\label{eq:models-last}
\strA \models_\inst{} |\alpha\rangle \Last
\end{equation}
to express a boolean query 
checking whether \textbf{there exists a successful execution} of $\alpha$ starting at the input structure $\strA$. Here, $\Last$ is an  expression for the last state of a path: 
$
\Last  :=  \ |any_M] \ \neg \mT,
$
where
$
any_M \ : = \ M_1 \sqcup  \dots  \sqcup   M_l
$
and  $M_1, \dots, M_l$ are all the atomic module names in  $\Sch$.  Inside the right-facing diamond modality $|\alpha\rangle$, we have, essentially, an imperative program.

\subsection{Expressing Programming Constructs}\label{sec:programming-constructs}
It is well-known that in Propositional Dynamic Logic  \cite{DBLP:journals/jcss/FischerL79}, imperative programming  constructs  are definable using a fragment of  regular languages, see the Dynamic Logic book by Harel, Kozen and Tiuryn \cite{HarelKozenTiuryn-DL}. A language that corresponds to the expressions composed of the imperative constructs is called \emph{Deterministic Regular (While) Programs} in  \cite{HarelKozenTiuryn-DL}.\footnote{	Please note that Deterministic Regular expressions and the corresponding Glushkov automata are unrelated to what we study here. In those terms, expression $a \comp a^*$ is Deterministic Regular, while $a^* \comp a$ is not. Both expressions are not in our language.}  Imperative  constructs  are definable in logic DetLDL$_f$ (see (\ref{eq:LLIF-two-sorted}) in Section \ref{sec:dynamic-logic}), as shown below.
$$
\begin{array}{l}
{\bf skip}\ : =\ \id,  \ \ \ \ 
{\bf fail}\ : =\ \rneg \id,\\
{\bf if} \ \phi \ {\bf then}\  \alpha \ {\bf else} \ \beta \ : = \ (\phi? \comp \alpha) \sqcup   \beta,\\

{\bf while} \ \phi \ {\bf do}\  \alpha : = (\phi? \comp \alpha)^\iter  \comp (\rneg \phi?),\\

{\bf repeat} \ \alpha \ {\bf until}\  \phi : = \alpha \comp ((\rneg \phi?) \comp \alpha)^\iter  \comp  \phi?.
\end{array} 	
$$
Thus,  importantly, the full power of regular languages is not needed to define these constructs.

\subsection{Computational Problem $\cP_{\alpha}$ Specified by $\alpha$}
\label{sec:computational-problem} We use (\ref{eq:models-last})  to formalize computational problems.
\begin{definition}	\label{def:comp-problem}
A \emph{computational problem $\cP_{\alpha}$ specified by $\alpha$} is an isomorphism-closed  class of  structures (set if the domain is fixed) for which query (\ref{eq:models-last}) evaluates to \emph{true}.

\begin{equation}\label{eq:comp-problem}
	\cP_{\alpha} \ : = \ \{ \strA \ \ \mid \ \ \strA \ \models_T\  |\alpha\rangle \Last\ [\e]   \}.
\end{equation}

\end{definition}

\noindent Formula $|\alpha\rangle \Last\ [\e]$ in (\ref{eq:comp-problem}) has a free function variable $\e$. As any formula with free variables,  it \emph{defines a relation}, here, a set of concrete Choice functions:

 \begin{equation}\label{eq:choice-witnesses}
		\ \ \    	H^{\alpha}_{\strA}  :=  \{ \ch \in \CH(\Un)	\ \mid \ \strA \models_\inst{} |\alpha\rangle \Last \ [\ch/\e] \},
	\end{equation}

\noindent The relation is a \textbf{set of certificates} for the membership $\strA \in \cP_{\alpha}$. 
For each concrete Choice function $\ch$, there is an associated set of finite $\ch$-paths, introduced in Definition \ref{def:h-path} at the end of the Technical Preliminaries section. For a finite input structure and a function $\ch$, there is a unique maximal $\ch$-path
	\begin{equation}
		\label{eq:string}
		\strA= \strA _0 \cdot \strA _1 \cdot \strA_2 \cdot \strA_3 \cdot \dots \cdot \strA_m = \last(\ppath)
	\end{equation}
that represents the trace of a computation of $\alpha$ with choices specified by $\ch$, as illustrated in Figure \ref{fig:tree}. Thus, we can view, equivalently,  maximal $\ch$-paths as certificates for (\ref{eq:comp-problem}). 
A connection between such certificates for $\alpha$ and those for its sub-expressions will be clear from the correspondence between Binary and Path semantics presented in Section \ref{sec:binary-path-correspondence}.

Our eventual goal is to analyze  witnesses in these two forms --- Choice functions and strings of structures --- in order to detect symmetries among computations. 
In Section \ref{sec:Path-Semantics}, we define a Path semantics for the Dynamic Logic (\ref{eq:LLIF-two-sorted}). The semantics establishes a precise connection between  specific Choice functions $\ch$ and  paths $\ppath$ in the form (\ref{eq:string}), and thus, between the two types of certificates.

\vspace{2ex}

\section{Complexity of the Main Task}
\label{sec:complexity-main-task}

In this section, we define our main computational task as the task of answering the query $\strA \models_\inst{} |\alpha\rangle \Last$ introduced in Section \ref{sec:unary-satisfaction}. We study a specific case of it, when the \emph{Choice function is fixed}. We analyze the data complexity \cite{Vardi82} of this task, where $\alpha$ and $T$ are fixed, and input $\strA$ varies.  We show that it is polynomial-time in the size of the input structure.

\subsection{Main Computational Task}
We define the main task we are going to study as follows.

\vspace{1ex}

\smallskip
\noindent \fbox{\parbox{\dimexpr\linewidth-2\fboxsep-2\fboxrule\relax}{  \textbf{Main Computational Task } {\ } \label{main-task}\\
		\underline{Given:} $\tau$-structure $\strA$, algebraic term $\alpha$, semantic context $\inst{}$  and Choice function $\ch$ \hspace{8ex}

		\vspace{1ex}

		\underline{Decide:}  
		\begin{equation}	\label{eq:main-task}	\strA \models_\inst{} |\alpha\rangle \Last \ [\ch/\e]
		\end{equation} }}

\vspace{1ex}

\noindent Taking into account the formalizations in Section \ref{sec:computational-problem}, it is clear that it the task of checking membership $\ch \in H^{\alpha}_{\strA}$ in the set of witnesses $H^{\alpha}_{\strA}$ is defined in (\ref{eq:choice-witnesses}), i.e., the task of  verifying whether a specific Choice function $\ch$ is a certificate of $\cP_{\alpha}$, the computational problem specified by $\alpha$, as introduced by (\ref{eq:comp-problem}) in Definition \ref{def:comp-problem}.
We study the above task for the logic that is defined as follows.

\begin{definition}\label{def:main-logic}
	\logic (or, simply, $\Lo$) is the logic  (\ref{eq:LLIF-two-sorted}) with {\rm SM-PP} atomic transductions. 
\end{definition}

Recall that PTIME is a class of all problems computable by a deterministic polynomial time Turing machine. 
We now show that the data complexity of the Main Task (\ref{eq:main-task}) for the logic $\Lo$ given by Definition \ref{def:main-logic} above is in PTIME. 

\begin{theorem} \label{th:membership-PTIME}
	For a \emph{fixed} Choice function, the data complexity of the Main Task (\ref{eq:main-task}), where the algebraic term, Choice function and semantic context  are fixed, and input structures vary,  for logic $\Lo$ is in \emph{PTIME}.
\end{theorem}

Note that there could be exponentially many such Choice functions $\ch$ in general. 
Thus, part (i) of Definition \ref{def:logic-captures-PRTIME} of a logic for PTIME is not yet satisfied. 
The eventual goal is to study symmetries among Choice functions, to determine  \emph{syntactic} conditions for when constructing just one witness is enough to determine membership in the computational problem $\cP_{\alpha}$ defined by (\ref{eq:comp-problem}).

To prove Theorem \ref{th:membership-PTIME}, we need the following notion.

\subsection{Register Kripke Structures}\label{sec:register-structures}
By the   semantics of (\ref{eq:one-sorted}) in terms of binary relations, we can interpret \logic over transition systems (Kripke structures) 
with the \textbf{universe $\Un$}, i.e., the set of $\tau$-structures  over the input domain $dom(\strA)$, and transitions given by  \textbf{binary relations on $\Un$}.
In \emph{Register Kripke  structures,} the ``dynamic'' part, i.e., all relations that change from state to state are monadic singleton-set, i.e., are similar to registers in computer programming.

\begin{definition}[Register Kripke structure]\label{def:Register-structure}
	A Kripke structure $\cR_{\strA}$ with states in $\Un$ and transitions in $\Un\times \Un$ is called \emph{Register} if it is a pointed Kripke structure rooted in $\strA \in \Un$,  where the only relations that change from state to state are unary singleton-set relations.
\end{definition}
We call such structures Register structures, by analogy with Register Machines \cite{ShepherdsonSturgis:1961}.  Thus, in a Register Kripke structure, we have a set of ``registers'', each containing exactly one domain element in each state. 
In a Register Kripke structure, each state can be viewed as a structure of the form:
\begin{equation}\label{eq:state-of-register-structure}
	\strB \ : = \    (dom({\strA});  \dots ,  (s(X_1))^\strB, \dots, (s(X_k))^{\strB}),
\end{equation}
where $X_i, \dots X_k$ are variables that represent registers. Note that the interpretations of the ``real'' (i.e., EDB in the database terminology) input variables in (\ref{eq:state-of-register-structure}) are omitted,  since these (potentially non-unary) input relations \emph{never change} as computation progresses. 
Alternatively, such a state is uniquely represented by tuple of domain elements 
$$
\ba \ := \ a_1, \dots, a_k,
$$
such that $a_i  \in (s(X_i))^\strB$.
There are $n^k$ such tuples, where $n = |dom(\strA)|$. This observation is formalized in the following lemma.
\begin{lemma}\label{lemma-Register-strructures-polynomial-number-of-states}
	If $\strA$ is an input structure, then a	Register Kripke structure $\cR_{\strA}$ has at most $n^k$ states, where $n = |dom(\strA)|$ and $k$ is the number of registers.
\end{lemma}

\begin{lemma}\label{lemma:OM-Datalog-Register}
	The interpretations of $\Lo$ can be Register Kripke structures only.
\end{lemma}
\begin{proof}
	By  Definition \ref{def:main-logic} of logic $\Lo$ (or \logic),  atomic transitions of $\Lo$ are Singleton-set Monadic  Primitive Positive definable. In DetLDL$_f$, due to the $\e$ operator, only one possible  interpretation of the output symbols is selected.  Thus, such modules produce only singleton-set relations on the output. The interpretation of the other symbols in $\tau$ is transferred by inertia from the previous state. Therefore, the resulting Kripke structures are Register.
\end{proof}

\subsection{Proof of Theorem \ref{th:membership-PTIME}}\label{sec:proof-membership-in-PTIME}
\begin{proof}
We need to argue that the operations of (\ref{eq:one-sorted}) preserve  PTIME data complexity of the main task.  
We analyze the complexity recursively, on the structure of an algebraic expression, and represent the algorithm as a set of rules of Plotkin's Structural Operational Semantics \cite{Plot81}. The transitions describe ``single steps'' of the computation, as in the computational semantics \cite{Henn90}.
 A primitive step  is either due to an atomic transition, or is (a subset of) the Identity  relation $\id$ on $\tau$-structures, i.e., is a ``test''.
  The execution of the algorithm consists of constructing a bottom-up  derivation tree, starting with an input structure and an algebraic term $\alpha$.	
	The goal is to apply the rules of the structural
	operational semantics (below) starting from $(\alpha, \strA)$ and arriving to $true$. 
	In that case, we say that 	the  \emph{derivation is successful}.
	
\vspace{2ex}
	
		\noindent	\textbf{Identity (Diagonal)}	$\id$:
	$$
	\frac{true}{(\id, \strB) \longrightarrow (\id, \strB)}  .
	$$

		\noindent	\textbf{Atomic modules}  $M(\underline{\bX};\bY)$, \textbf{with} $\varepsilon$: 
		$$
		\frac{true}{(\varepsilon M(\underline{\bX};\bY), \strB_1) \longrightarrow (\id, \strB_2)}  \mbox{ if }    \strB_2|_{(\tau \setminus s(\bY))}=
		\strB_1|_{(\tau \setminus s(\bY))},
		$$
	$\strB_2 \in \inst(M(\underline{\bX};\bY))$ and	$\strB_2$ is the output of $\e M(\underline{\bX};\bY)$ according to the  Choice function $\ch$ that instantiates  $\varepsilon$.

\vspace{1ex}
		
		\noindent	\textbf{Sequential Composition} $\alpha \comp \beta$:
			$$
		\frac{(\alpha, \strB_1) \longrightarrow (\alpha',\strB_2) }{(\alpha \comp \beta, \strB_1) \longrightarrow
			(\alpha' \comp \beta, \strB_2) },
		$$
		$$
		\frac{(\beta, \strB_1) \longrightarrow (\beta',\strB_2)  }{(\id \comp \beta, \strB_1) \longrightarrow
			(\id \comp \beta', \strB_2) }.	
		$$

\vspace{1ex}

		\noindent	\textbf{Preferential Union}  
			$\alpha \sqcup \beta$:
	$$
	\frac{(\alpha, \strB_1) \longrightarrow (\alpha',\strB_2)   }{(\alpha \sqcup  \beta, \strB_1) \longrightarrow
		(\alpha' , \strB_2) }.
	$$
	That is, $\alpha \sqcup  \beta$  evolves according to the instructions of  $\alpha$, if  $\alpha$ can successfully evolve to $\alpha'$. 
	$$
	\frac{(\beta, \strB_1) \longrightarrow (\beta',\strB_2) \mbox { and  } (\rneg \alpha, \strB_1) \longrightarrow (\rneg \alpha,\strB_1)    }{(\alpha \sqcup  \beta, \strB_1) \longrightarrow
		( \beta', \strB_2) }  .
	$$
The rule says that $\alpha \sqcup  \beta$  evolves according to the instructions of  $\beta$, if $\beta$ can successfully evolve, while $\alpha$ cannot.		
		
	\vspace{1ex}	
		
			\noindent	\textbf{Unary Negation (Anti-Domain)} $\rneg \alpha$:
	There are no one-step derivation rules for $\rneg \alpha$. Instead, we try, step-by-step, to derive $\alpha$, and if not derivable, make the step.	
					$$
			\frac{true    }{(\rneg \alpha, \strB_1) \longrightarrow ( \id,\strB_1)  } \mbox{ if derivation of $\alpha$ in $\strB_1$ fails} .
			$$		
	\noindent	\textbf{Left Unary Negation (Anti-Image)} $\lneg \alpha$:	
			$$
		\frac{true    }{(\lneg \alpha, \strB_2) \longrightarrow ( \id,\strB_2)  } \mbox{ if derivation of $\alpha^\converse$ in $\strB_2$ fails} .
		$$

\vspace{1ex}

	\noindent	\textbf{Right Selection}	 	
		$$
 \frac{true  }{(\sigma^r_{ X\eq  Y} (\alpha), \strB_1) \longrightarrow
 	(\id , \strB_2) } \mbox{ if }  (s(X))^{\strB_2} =  (s(Y))^{\strB_2}
 $$
 and the derivation for $\alpha$ succeeds.

\vspace{1ex}

	\noindent	\textbf{Converse}
		We assume that, at a preprocessing stage,  Converse  is pushed inside the formula as much as possible. 
		Thus, we need to deal with the atomic case only.
		$$
	\frac{true}{((\varepsilon M(\underline{\bX};\bY))^\converse, \strB_2)\! \longrightarrow \! (\id, \strB_1)} \mbox{ if }    \strB_2|_{(\tau \! \setminus\! s(\bY))}\!\!=\!
	\strB_1|_{(\tau \! \setminus\! s(\bY))},
	$$
	$\strB_2 \in \inst(M(\underline{\bX};\bY))$ and	$\strB_2$ is the output of $\e M(\underline{\bX};\bY)$ according to the fixed Choice function that instantiates  $\varepsilon$.

\vspace{1ex}

	\noindent	\textbf{Maximum Iterate} $\alpha^\iter$:
		$$
	\frac{(\alpha, \strB_1) \longrightarrow (\alpha',\strB_2)}{(\alpha^\iter, \strB_1) \longrightarrow
		(\alpha' \comp \alpha^\iter, \strB_2) }.
	$$
	Thus, Maximum Iterate evolves according to $\alpha$, if $\alpha$ can evolve successfully. It fails if $\alpha$ cannot evolve.

\vspace{3mm}

The derivation process is deterministic: whenever there are two rules for a connective, only one of them is applicable. The process always terminates: the number of structures is finite, and the  program (algebraic term) on top of each rule is always smaller, in terms of the number of remaining actions to execute, than the one on the bottom.
The correctness of the algorithm follows by  induction on the structure of the algebraic expression, using the semantics in terms of binary relations.

 We are now ready to present our argument.
A ``yes'' answer to the main computational  task  (\ref{eq:main-task}) is witnessed by 
a word in the alphabet $\Un$  that starts from the input structure $\strA$, or, equivalently, \emph{by a successful derivation} using the rules above.  
We argue, by induction on the structure of the algebraic term $\alpha$ in (\ref{eq:one-sorted}), that the complexity of the task (\ref{eq:main-task})  is in PTIME. 
Query evaluation  for SM-PP atomic modules  is computable in uniform AC$^0$, and thus, in polynomial time in the size of the input structure $\strA$. Only one output is selected (externally) by the fixed interpretation of  $\varepsilon$ at no cost.
We proceed by induction. 
As is seen from the rules, the non-recursive operations, Sequential Composition, Preferential Union, Unary Negations and the Right Selection, preserve PTIME evaluations with respect to data complexity.
For the case of the limited recursion construct, the derivation for  $\alpha ^\iter$ always terminates  because (a) the transition system generated by executing the program is  finite (for a finite input domain); and (b)
the process stops when there is no outgoing $\alpha$-transition, which is the base case of $\alpha ^\iter$.  Thus, Maximum Iterate explicitly disregards cycles -- if there is an infinite loop, there is no model. Moreover, 
since, by Lemma \ref{lemma:OM-Datalog-Register} 
		the Kripke structures we obtain 
are Register, by Lemma \ref{lemma-Register-strructures-polynomial-number-of-states}, they have at most $n^k$ states.  
 It implies that a trace produced by Maximum Iterate can be of length at most $n^k$, the total number of states in a Register transition system.
 For a fixed formula, there could be only a finite number of Maximum Iterate constructs. Since each step takes at most a polynomial time, the total data complexity of computing the main task is polynomial. \end{proof}

\vspace{2ex}

\section{Path Semantics}
\label{sec:Path-Semantics}

 In this section, we develop a different perspective on the algebra. 
 We define an alternative (and equivalent) \emph{path} semantics of (\ref{eq:LLIF-two-sorted}) initially inspired by \cite{DeGiacomo-Vardi:IJCAI:013}.  This semantics provides a link to automata theory and the notion of formal language equivalence, $\cL(\alpha) = \cL(\beta)$. Moreover, it allows us to view our algebra as an algebra of  partial functions on finite words in $\Un^+$.

\subsection{Path Semantics}\label{sec:path-semantics}
Recall that a \emph{path} is  a finite word $\ppath \subseteq \Un^+$, as defined in Section \ref{sec:Preliminaries}. We use $\ppath(i,j)$ to denote the \emph{subword} obtained from $\ppath$ starting from position $i$ and terminating in position $j$, where $0\leq i \leq j \leq \last(\ppath)$.

Path semantics of the two-sorted syntax (\ref{eq:LLIF-two-sorted})  is defined by a simultaneous induction on formulae $\phi$ and process expressions $\alpha$.

\noindent \textbf{ (a) State Formulae (line 2 of (\ref{eq:LLIF-two-sorted}))} 
We define when  state $\ppath(i)$ of path $\ppath$ \emph{satisfies} formula $\phi$  under  $\inst{}$, denoted $ \ppath, i \models_\inst{} \phi$,  where  $0 \leq i \leq |\ppath|$, as follows:
\vspace{1ex}
\begin{itemize}
	\item $\ppath,   i \models_\inst{} M_p(\underline{\bX};\ )$, for $M_p \in \Sch{}$, iff $\ppath(i) \models_\inst{} M_p(\underline{\bX};\ )$, for  proposition module $M_p$.

	\item $ \ppath,  i \models_\inst{} \neg \phi$ iff $   \ppath, i \not \models_\inst{} \ \phi$,

	\item $ \ppath,   i \models_\inst{}  | \alpha \rangle \phi$ iff  there is  $j$,  $ i\leq j \leq \last(\ppath)$, with $\ppath(i,j) \in {\cL} (\alpha)$ such that $\ppath(j) \models \phi$,

	\item $\ppath,    i \models_\inst{}  | \alpha ] \phi$ iff  for all   $j$,  $ i\leq j \leq \last(\ppath)$, such that $\ppath(i,j) \in {\cL} (\alpha)$, we have $\ppath(j)\models \phi$.

	The cases for $ \langle \alpha | $ and $ [ \alpha | $ are similar. 	
\end{itemize}

\noindent \textbf{\bf (b) Process (Path) Formulae (line 1 of (\ref{eq:LLIF-two-sorted}))} 
These formulae are interpreted  as in the binary semantics, but always with respect to a path.  
For each $\alpha$,  relation $\ppath(i,j) \in {\cL} (\alpha)$ is defined as follows.

\begin{itemize}
	
	\item Atomic modules:
	$\ppath (i,i+1) \in \cL(\varepsilon M_a(\underline{\bX};\bY)) $ if $  (\ppath(i),\ppath(i+1))  \in  \sem{ M_a(\underline{\bX};\bY) }$ and  $M_a\in \Sch{} $, for action module $M_a$.

	\item Diagonal: 
	$\ppath (i,j) \in \cL(\id) $ if $j=i$.
	
	\item Test: $\ppath (i,j) \in \cL(\phi?) $ if $j=i$ and $  \ppath, i \models_\inst{} \phi$.
	
	\item Right Unary Negation: $\ppath (i,j) \in \cL(\rneg \alpha) $ if $j=i$ and there is no $k$ such that 
	$\ppath (i,k) \in \cL(\alpha)$.
	
	\item Left Unary Negation: $\ppath (i,j) \in \cL(\lneg \alpha) $ if $j=i$ and there is no $l$ such that 
	$\ppath (l,j) \in \cL(\alpha)$.

	\item Sequential Composition: $\ppath (i,j) \in \cL(\alpha \comp \beta )  $ if there exists $l$, with $i\leq l \leq j$, s.t. $\ppath (i,l) \in \cL(\alpha  ) $ and  $\ppath (l,j) \in \cL(\beta )$.
	
	\item Preferential Union:  $\ppath (i,j) \in \cL(\alpha \sqcup \beta )  $ if   (recall that $ \Dom (\alpha) := \rneg \rneg \alpha $):
	$$
\!\!\!\!	\left\{   \! \! \! \left.
	\begin{array}{l}
	\ppath (i,j) \! \in \! \cL(\alpha)  \   \mbox{ and } \ \ppath(i) \models_\inst{} \Dom(\alpha),  \mbox{ or }\\
	\ppath (i,j) \! \in \! \cL(\beta  ),    \ \ \ppath(i) \not \models_\inst{} \Dom(\alpha)     \mbox{ and  }  \ppath(i) \models_\inst{} \Dom(\beta)     
	\end{array} \right.\! \! \! \right\}.
	$$

	\item 	Maximum Iterate: $\ppath (i,j) \in \cL(\alpha^\iter )  $ if $i=j$ and $ \ppath (i,j) \in \cL( \rneg \alpha ) $ or there exists $l$, with $i\leq l \leq j$, such that $\ppath (i,l) \in \cL(\alpha  ) $ and $\ppath (l,j) \in \cL(\alpha^\iter )$.

	\item Right Selection:  $\ppath (i,j) \in \cL(\sigma^r_{ X \eq  Y} (\alpha)) $ if $\ppath (i,j) \in \cL(\alpha)  $ and 	$	   (s(X))^{\ppath(j)} =  (s(Y))^{\ppath(j)}	$,
	
	\item Converse:
	$\ppath (i,j) \in \cL(\alpha^\converse)  $ if
	$(\ppath (i,j))^{-1} \in \cL(\alpha)  $. 
	
\end{itemize}

\subsection{Correspondence between Binary and Path Semantics}\label{sec:binary-path-correspondence}

Recall the notion of an $\ch$-path introduced at the end of  the Technical Preliminaries (Section \ref{sec:Preliminaries}). The following theorem shows  a connection between path and binary semantics for \emph{process} formulae. 
\begin{theorem} \label{th:path-binary-semantic}
	$\ppath(i,j) \in {\cL} (\alpha) \ \  \Leftrightarrow\ \  (\ppath(i),\ppath(j))  \in \sema {\alpha}$, where  $\ppath$ is an $\ch$-path, i.e., $\ppath\! \in\! b(\ch)$ for a concrete Choice function $\ch$.
\end{theorem}	
Since $b(\ch)$ determines an equivalence class of Choice functions,  $\ppath\! \in\! b(\ch)$ and $\strA=\ppath(0)$ for an input structure $\strA$, strings of structures (\ref{eq:string}) can be viewed as witnesses for the computational problem (\ref{eq:comp-problem}) introduced in Definition \ref{def:comp-problem}, as already mentioned in Section \ref{sec:computational-problem}.
\begin{proof} We show the equivalence by induction on the structure of the process formulae in (\ref{eq:LLIF-two-sorted}).  We demonstrate the existence of $\ch$ by constructing  a subset $b$ of $b(\ch)$. Any extension of this subset to the full $b(\ch)$ corresponds to an equivalence class $[\ch]$ of Choice functions, see Figure \ref{fig:tree}. Any function in that class gives us the  $\ch$ needed in the theorem.

\vspace{1ex}

\noindent \underline{Base cases:}  
	\begin{itemize}
		\item Atomic modules:		$\ppath (i,i+1) \in \cL(\varepsilon M_a(\underline{\bX};\bY)) $ \\
	 $\!\!\Leftrightarrow$ (by the path semantics)  $\!(\ppath(i),\ppath(i+1))  \! \in \! \sem{ M_a(\underline{\bX};\bY) }$

	$\Leftrightarrow\ \mbox{ (by the binary semantics (\ref{eq:epsilon-semantics})) }$
	$(\ppath(i),\ppath(i+1)) \in \sema {\varepsilon M_a(\underline{\bX};\bY) }$, where  $\ppath_i,\ppath_{(i+1)}\in b\subset  b(\ch)$ (recall that $\ppath_i$ is the prefix of $\ppath$ ending in position $i$). That is, $\ch$ must ``agree'' with  $M_a$, $f(\ppath_i)= \ppath(i+1)$.  
	
	\item Diagonal: 
	$\ppath (i,j) \in \cL(\id) $ $\Leftrightarrow\ \mbox{ (by the path semantics) } $ $j=i$
	
		$\Leftrightarrow\ \mbox{ (by the binary semantics) }$
	$(\ppath(i),\ppath(i)) \in \sema {\id }$. That is, only $\ppath_i$ is added to $b\subset  b(\ch)$.

	\item Test: $\ppath (i,j) \in \cL(\phi?) $ $\Leftrightarrow\ \mbox{ (by the path semantics) } $ $j=i$ and $  \ppath, i \models_\inst{} \phi$

		$\Leftrightarrow\ \mbox{ (by the binary semantics) }$
	$(\ppath(i),\ppath(i)) \in \sema {\phi? }$.
	Again, only $\ppath_i$ is added to $b\subset b(\ch)$.

\end{itemize}

\noindent \underline{Inductive cases:}  Assume that, for the paths  $\ppath (i,l)$ and $\ppath (k,j)$, the equivalence holds, and $\ppath_i$, $\ppath_l$, $\ppath_k$, and $\ppath_j$ (by construction, all prefixes of the same path $\ppath$)  are in  $b\subset b(\ch)$, for some $\ch$. 
We show that the theorem holds for all applications of the algebraic operations.

\begin{itemize}	
	
	\item Sequential Composition: $\ppath (i,j) \in \cL(\alpha \comp \beta )  $ $\Leftrightarrow\ \mbox{ (by the path semantics) } $  there exists $l$, with $i\leq l \leq j$, such that $\ppath (i,l) \in \cL(\alpha  ) $ and $\ppath (l,j) \in \cL(\beta )$. By inductive hypothesis, 
	$(\ppath(i),\ppath(l)) \in \sema {\alpha} $, $(\ppath(k),\ppath(j)) \in \sema {\beta}$, where $l=k$, and $\ppath_i$, $\ppath_l$ and $\ppath_j$   are in  $b\subset b(\ch)$, for some $\ch$. By the binary semantics of Composition, $(\ppath(i),\ppath(j)) \in \sema {\alpha \comp \beta} $. Since  $\ppath_i$ and $\ppath_j$   are already in  $b\subset b(\ch)$, the theorem holds for this case. 	
\end{itemize}
For the other inductive cases, the same type of reasoning applies. We conclude that the theorem holds.
\end{proof}

The theorem implies that the \emph{binary} satisfaction relation (\ref{eq:binary-sat}) of the semantics in terms of binary relations and the path semantics of \emph{process} formulae of (\ref{eq:LLIF-two-sorted}) are equivalent.

\noindent We now show the  correspondence for the  \emph{unary} satisfaction relation (\ref{eq:unary-sat}) and the path semantics of \emph{state} formulae of (\ref{eq:LLIF-two-sorted}).

\begin{theorem}
	For any state formula $\phi$ of (\ref{eq:LLIF-two-sorted}), we have:
$
\ppath, 0 \models_\inst{} \phi    \ \  \Leftrightarrow\ \ \strA \models_\inst{} \phi[\ch/\e],
$ 	
 where $\ppath$ is an  $\ch$-path and $\ppath(0) =\strA$.	 
\end{theorem}
\begin{proof}
	The property follows by induction on the structure of the algebraic expressions of (\ref{eq:LLIF-two-sorted}). 
	
	\vspace{1ex}
	
\noindent \underline{Base case, 	atomic modules-propositions:}  

	$\ppath,   0 \models_\inst{} M_p(\underline{\bX};\bX)   \ \ \\ \Leftrightarrow\ \mbox{ (by the path semantics) } \  \ppath(0) \models_\inst{} M_p(\underline{\bX};\bX)  $
\\
	$  \Leftrightarrow $ (by Theorem \ref{th:path-binary-semantic}, assuming  $\ppath$ is an $\ch$-path) $ (\strA, \strA ) \in  \sema{ M_p(\underline{\bX};\bX)  } $	\\
$\ \  \Leftrightarrow\ \ $ (by (\ref{eq:binary-sat}) and (\ref{eq:unary-sat}) )
	$   \strA \models_\inst{} M_p(\underline{\bX};\bX) [\ch/\e]  $. 
	
	\noindent Choice function does not play any role in the case of atomic modules-propositions.

\noindent \underline{Case $\neg \phi$:} 

$ \ppath,  0 \models_\inst{} \neg \phi\ \  \Leftrightarrow\ \mbox{ (by the path semantics) }  \  \ \ppath, 0 \not \models_\inst{} \ \phi \ \   \ $ 
\\
$\ \  \Leftrightarrow\ \ $ (by Theorem \ref{th:path-binary-semantic}, assuming  $\ppath$ is a $\ch$-path)
\ \ $ (\strA, \strA ) \not\in  \sema {\phi  }$

\noindent $\Leftrightarrow\ \ $ (by (\ref{eq:binary-sat}), (\ref{eq:unary-sat}) and the binary semantics of Unary Negation 
by recalling $ \neg  \phi : = \rneg \phi$)   
\ \ $ \strA \models_\inst{}  \neg \phi [\ch/\e ]$.

\noindent \underline{Case $| \alpha \rangle \phi$:} 

 $ \ppath,   0 \models_\inst{}  | \alpha \rangle \phi$ \\
 $   \Leftrightarrow\ \ $ (by definition of the path semantics of $| \alpha \rangle \phi$)
 \hspace{-10mm}

\noindent $\exists\, j  \big( 0\!\! \leq\! j\! \leq \last(p) \mbox{ such that } p(0,j)\!\! \in \!\! {\cal L}(\alpha) \mbox{ and } \ppath(j) \models_\inst{} \phi \big)$ 

\noindent $\Leftrightarrow$ (by Theorem \ref{th:path-binary-semantic}, assuming  $\ppath$ is an $\ch$-path)\\
  $\exists j\ \big((\strA, \strB ) \in  \sema{ \alpha }   \mbox { and } (\strB, \strB ) \in  \sema{ \phi }, \mbox{ where  $\ppath(j)= \strB$}  \big)$

\noindent $\Leftrightarrow$ (by binary semantics of Composition)\\
 $\exists \strB \ \big((\strA, \strB ) \in   \sema{ \alpha \comp \phi }   \big)$ \\
   $\ \  \Leftrightarrow\ \ $ (by binary semantics of Domain operation)\\
    $   (\strA, \strA ) \in  \sema{ \Dom(\alpha \comp \phi)  } $ \\
$\ \  \Leftrightarrow\ \ $ (by binary semantics of $| \alpha \rangle \phi$) \\   \ \
 $ \strA \models_\inst{}  | \alpha \rangle \phi [\ch/\e]$.
	
		The cases for $| \alpha ]$,  $ \langle \alpha | $ and $ [ \alpha | $ are similar.	
\end{proof}

\vspace{2ex}

\section{Examples}\label{sec:Examples}

In this section, we first present some examples of reachability and counting properties, expressed without the use of a counting construct. Then we show an example where two fundamentally different kinds of propagations are combined, with the help of the Choice construct. Towards the end of this section, we demonstrate how the operations of Regular Expressions can be mimicked in the logic. 

\subsection{Combining Counting and Reachability}

It is known that fundamental reachability and counting properties are decidable in PTIME.

\vspace{2ex}
\begin{center}{\small
	\begin{tabular}{l}
		{\textbf{Reachability}}\\
		\toprule
		s-t-Connectivity\\
		Transitive Closure \\
		Same Generation \\
		\dots
		\\
		
	\end{tabular} \hspace{5ex} 
	\begin{tabular}{l}
		{\textbf{Counting}}\\
		\toprule
		EVEN\\
		Equal Size\\
		Size 4 (or 7, or 16)\\
		\dots\\
\end{tabular}}
\end{center}

\noindent Counting properties include, for example,  query EVEN (which asks whether the cardinality of a given structure is even), Equal Size (which compares the cardinalities of two input sets). Reachability properties include   s-t-Connectivity, Transitive Closure  and other similar properties. 
Yet, reachability and cardinality properties seem to reside in two different worlds.
On one hand,  on unordered structures, the usual logics and query languages with an iteration construct do not have the ability to count. 
For instance, the query EVEN  is computable in LOGSPACE but cannot be expressed in the transitive closure logics FO(TC) and FO(DTC), or in fixed-point logics FO(LFP) and FO(PFP). Even $\cL^{\omega}_{\infty,\omega}$,  an infinitary logic that is able to express a very  powerful form of iteration, is unable to count.\footnote{Since the query EVEN is in PTIME, by Fagin's  theorem \cite{Fagin:1974}, it can be expressed in second-order existential logic. However, it cannot be expressed in the monadic second-order logic, for which the separation of its existential and universal fragments is known, unless an order on domain elements is given. }
On the other hand, adding more and more  expressive counting constructs to first-order logic, on unordered structures, does not produce the ability to express transitive closure.
For instance,  despite its enormous counting power,  $\cL^*_{\infty,\omega}({\bf Cnt})$ (an extremely powerful infinitary counting logic, cf. Chapter 8 of \cite{Libkin-book04})  remains local and cannot express properties that require iterative computations (in logic, ``local'' means that the truth of a formula can be determined by looking at neighbourhoods of specific radii, cf. \cite{Libkin-book04}).

So, on unordered structures, the two concepts, cardinality and reachability, seem to be orthogonal. However, combining iteration and counting \textbf{directly} is not enough: it has been known for almost thirty years that FO(FP)+C does not capture all of PTIME, as shown by Cai, F\"{u}rer and Immerman \cite{CFI92}. Their example, now referred to as the CFI property, is not expressible in FO(FP)+C.

Next, we show how the orthogonality is eliminated with the help of our Choice construct. We present two kinds of examples, cardinality and reachability, and then an example where two types of propagations are combined.
The reader will notice that, in modelling these polynomial time properties, the exact sequence of choices does not matter.

To proceed to our examples, we need some definable modalities to refer back in the computation.
Recall that we defined
$
any_M \ : = \ M_1 \sqcup  \dots  \sqcup   M_l,
$
where $M_1, \dots, M_l$ are all the atomic module names in  $\Sch$.
This notation will now be used to introduce abbreviations for temporal properties.
In our setting, Exists ($\E$)  and  Globally ($\G$) are modal operators analogous to $F$ (Exists) $G$ (Globally)  of Linear Temporal Logic (LTL), respectively.
We also introduce Back Exists ($\BE$)  and Back Globally ($\BG$), which are similar to $P$  and $H$, respectively,  of  Arthur Prior's Tense Logic \cite{Prior:1957,Prior:1967} and a version of LTL with backwards-facing temporal modalities.
We define:
$
\E \phi  :=  | {\bf repeat \ } 
any_M 
{\ \bf until } \ \phi \rangle \, \mT .
$
Similarly, for reasoning ``backwards in time'', we have:
$
\BE \phi  :=  \langle ({\bf repeat \ } 
any_M  
{\ \bf until } \ \phi)^\converse | \, \mT .
$
The program executes the actions of $\alpha$ ``backwards'' by ``undoing'' them  until $\phi$ is found to be true.
The dual modalities are:
$
\begin{array}{l}
\G \phi \ := \ \rneg \E \rneg \phi,\ \ \ \ \ 
\BG \phi \ := \ \lneg \BE \lneg \phi.
\end{array}
$
Thus, in particular, $
	\BG \phi  :=  [ ({\bf repeat \ } 
	any_M  
	{\ \bf until } \ \phi)^\converse | \, \mT .
	$

\vspace{2ex}

To avoid a notational overload, we use $\strA$ to refer to both a $\tau$-structure, which is an initial state in the transition system, and to the interpretation of the input vocabulary (a subset of $\tau$), e.g., a graph. The assumption is that, outside of that input graph, $\strA$ interprets all predicate symbols in $\tau$ arbitrarily.

\subsection{Cardinality Examples}

 \medskip

\noindent \fbox{\parbox{\dimexpr\linewidth-2\fboxsep-2\fboxrule\relax}{ Problem: \textbf{EVEN} $\alpha_E $
		
		Given: A structure $\strA$ with an empty vocabulary. \\
 		Question: Is $|dom(\strA)|$ even?}}%
 \medskip 
  
  EVEN is PTIME computable, but is not expressible in MSO, Datalog, or any fixed point logic, \emph{unless} a linear order on the domain elements is given. 
We construct a 2-coloured path in the transition system using $E$ and $O$ as labels. 
Due to using a $\BG$ check, the elements in the path never repeat. Each trace gives us an implicit linear order on domain elements.   For  deterministic atomic modules, we omit the Epsilon operator.
$$
 \begin{array}{lcl}
 \GuessP & : = &  \varepsilon\  \defin{  P(x) \rul \ \ \ },
 \\
 \CopyPO & : = & \defin{O(x ) \rul \underline{P(x)}},\\
 \CopyPE & : = & \defin{ E(x ) \rul \underline{ P(x)}}.
 \end{array}
 $$
$$
 \begin{array}{l}
 \hspace{-3mm}\GuessNewO  : = 
 \big(\GuessP  \comp   \BG (\sigma^r_{P \not \eq  E }    (\sigma^r_{P \not \eq O } (\id)))\big) \comp  \CopyPO,\\
\hspace{-3mm}\GuessNewE  : = 
 \big(\GuessP  \comp   \BG (\sigma^r_{P \not \eq E }    (\sigma^r_{P \not \eq O } (\id)))\big) \comp  \CopyPE.
 \end{array}
 $$
Here, $\sigma^r_{P \not \eq O}(\id)$  is an abbreviation for $\rneg\!\!\!\sigma^r_{P  \eq O}(\id)$.
 The problem EVEN is now specified as:
 $ 
 \alpha_E:= (\GuessNewO \comp  \GuessNewE ) ^\iter \comp \rneg \GuessNewO.
 $
 The program is successfully executed if each chosen element is  different from any elements selected so far in the current sequence of choices, and if $E$ and $O$ are guessed in alternation.

While $E$ and $O$ are singletons, their interpretations change  throughout the computation, the same way the content of registers changes in Register machines. Thus, the computed relations (called IDB in the Database community) are accumulated in the path.

 The expression $\alpha_E$ does not depend on an input vocabulary. 
 Given a structure $\strA$ over an empty  vocabulary, 
 $
 \strA \models_T | \alpha_E \rangle \Last,
 $
 holds whenever there is a successful execution of  $\alpha_E$, that is, the size of the input domain is even.

 \vspace{3mm}

 \noindent \fbox{\parbox{\dimexpr\linewidth-2\fboxsep-2\fboxrule\relax}{ Problem: \textbf{Size Four} $\alpha_4$
 		
 		Given: A structure $\strA$ with an empty vocabulary. \\
 		Question: Is $|dom(\strA)|$ equal to 4?}}
 \smallskip 
 $$
 \begin{array}{lcl}
 \GuessP & : = &  \varepsilon\  \defin{  P(x) \rul \ \ \ },\\
 \CopyPQ & : = & \defin{Q(x ) \rul \underline{P(x)}}.
 \end{array}
 $$
 $$
 \begin{array}{l}
 \GuessNewP  : = 
 \big(\GuessP  \comp   \BG (\sigma^r_{P \not \eq Q }     (\id))\big)\ \comp \ 
\CopyPQ.
 \end{array}
 $$
 The problem Size Four is now axiomatized as:
 $$ 
 \begin{array}{r}
 \alpha_4:= \GuessNewP^4   \comp \rneg \GuessNewP,  
  \end{array}
 $$
where the power four means that we execute the guessing procedure four times sequentially.
 The answer to the question $
 \strA \models_T | \alpha_4 \rangle \Last,
 $
 is yes, if and only if the input domain is of size 4. Obviously, such a program can be written for any natural number.

\vspace{2mm}
 
 \noindent \fbox{\parbox{\dimexpr\linewidth-2\fboxsep-2\fboxrule\relax}{ Problem: \textbf{Same Size} $\alpha_{\rm eq\_size}(\underline{P},\underline{Q})$
 		
 		Given: Two unary relations P and Q. \\
 		Question: Are P and Q of the same size?}}%
 \medskip

 We pick, simultaneously,  a pair of elements  from the two input sets, respectively: 
 $$
 \begin{array}{lcl}
 \PickPQ & : = & \varepsilon \  \defin{ \Pick_P(x ) \rul \underline{P(x)}
 ,\\ \Pick_Q(x) \rul \underline{Q(x)}
 }.
 \end{array}
 $$
Store the selected elements temporarily:
$$
\begin{array}{lcl}
\Copy & : = &   \defin{ P'(x ) \rul \underline{\Pick_P(x )},\\
 Q'(x ) \rul \underline{\Pick_Q(x )}
}.
\end{array}
$$
$$
\begin{array}{l}
\GuessNewPair  \ : = \\
\big(\PickPQ  \comp   \BG (\sigma^r_{Pick_P \not \eq  P' }  (\sigma^r_{\Pick_Q \not \eq  Q' }    (\id))\big)\ \comp \  \Copy.
\end{array}
$$
The problem Same Size is now axiomatized as:
$
\alpha_{\rm eq\_size} (\underline{P},\underline{Q}):=  (\GuessNewPair) ^\iter \comp \rneg \PickPQ.
$
The answer to the question $
\strA \models_T | \alpha_{\rm eq\_size} \rangle \Last ,
$
is yes, if and only if the extensions of predicate symbols in the input structure $\strA$ are of equal size. 
 
 \subsection{Reachability Examples}

\noindent
\fbox{\parbox{\dimexpr\linewidth-2\fboxsep-2\fboxrule\relax}{ 
		Problem: \textbf{s-t-Connectivity} $\alpha(\underline{E},\underline{S},\underline{T})$\\
		Given: Binary relation $E$,  two constants $s$ and $t$ represented as singleton-set relations $S$ and $T$. \\
		Question: Is $t$ reachable from $s$ by following the edges?}}
\medskip  

We use the definable  constructs of imperative programming. 
$$
\begin{array}{l}
\alpha(\underline{E},\underline{S},\underline{T}) :=  \ \  M_{base\_case}; \\
{\bf repeat \ }   \big(   M_{ind\_case}  \comp \  \BG (\sigma _{\Reach' \not \eq \Reach }  (\id) ) \big)  \comp \Copy  \\
 \hspace{1cm}   \ \ {\bf until } \ \sigma^r_{\Reach \eq T}(\id) .
\end{array}
$$
On the input, we have a binary relation $E$,  two constants $s$ and $t$ represented as singleton-set relations $S$ and $T$. Inputs $E$, $S$ and $T$ are underlined.
We  use a unary relational variable $\Reach$. Initially, the corresponding relation contains the same node as $S$. The execution is terminated when $\Reach$ equals $T$. Variable $\Reach'$ is used as a temporary storage.
To avoid guessing the same element multiple times, we use the (definable) Back Globally ($ \BG$) modality. 
Atomic modules are axiomatized in non-recursive Datalog with monadic output predicates as follows:
$$
\begin{array}{l}
\begin{array}{lcl}
M_{base\_case}& : = &  \  \defin{   \Reach(x ) \rul \underline{S(x)}},
\end{array}\\
\begin{array}{l}
M_{ind\_case}  : =  \varepsilon
\  \defin{   \Reach'(y ) \rul \underline{\Reach(x)}, \underline{E(x,y)}},
\end{array}\\
\begin{array}{l}
\Copy \ \ : = 
{\ \ \  }  \ \defin{  \ \Reach(x ) \rul \underline{\Reach'(x)}}.
\end{array}
\end{array}
$$
Here, module $M_{ind\_case}$  is the only non-deterministic module. For the other two modules, $\varepsilon $ in the expression for $\alpha(\underline{E},\underline{S},\underline{T})$ is omitted because they are deterministic (i.e., the corresponding binary relation is a function).
Input symbols of each atomic module are underlined.
Given structure $\strA$ over a vocabulary that matches the input variables $E,S$ and $T$, including matching the arities,   by checking 
$
\strA \models_T |\alpha \rangle \Last,
$
we verify that there is a successful execution of  $\alpha$. That is, $t$ is reachable from $s$ by following the edges of the input graph.

\vspace{2mm}

\noindent \fbox{\parbox{\dimexpr\linewidth-2\fboxsep-2\fboxrule\relax}{

		{Problem: \textbf{Same Generation}  $\alpha_{\rm SG}(\underline{E},\underline{\Root}, \underline{A},\underline{B})$}
		
		Given: Tree -- edge  relation: $E$;  root:   $\Root$;  two nodes represented by unary singleton-set relations:   $A$ and $B$ 
		
		Question: Do $A$ and $B$ belong to the same generation in the tree? }} 

\medskip

Note that,  since we do not allow binary \emph{output} relations (binary extensional variables are allowed),  we need to capture the notion of being in the same generation through coexistence in the same structure. 
		$$
		\begin{array}{lcl}
		M_{base\_case} & : = & \defin{ \Reach_A(x ) \rul \underline{A(x)},\\
			\Reach_B(x ) \rul \underline{B(x)}
		}.
		\end{array}
		$$
		We do a simultaneous propagation starting from the two nodes:
		$$
		\begin{array}{l}
		\hspace{-2mm} M_{ind\_case}  : =  
		\ \varepsilon \ \defin{  \Reach_A'(x ) \rul \underline{\Reach_A(y)}, \underline{E(x,y)},\\
			\Reach_B'(v ) \rul \underline{\Reach_B(w)}, \underline{E(v,w)} }.
		\end{array}
		$$
		 This atomic module
		specifies that, if elements $y$ and $w$, stored  in the interpretations of $\Reach_A$ and $\Reach_B$ respectively, coexisted in the previous state,  then $x$ and $v$ will coexist in the successor state. We copy the reached elements into ``buffer'' registers:
		$$
		\begin{array}{l}
		Copy  : =   
		\defin{ \Reach_A(x ) \rul \underline{\Reach_A'(x )},\\
			\Reach_B(x ) \rul \underline{\Reach_B'(x )}}.
		\end{array}
		$$
   The resulting interpretation of $\Reach_A$ and $\Reach_B$ coexist in one Tarski (first-order) structure, which is a state in a higher order Kripke structure.
The algebraic expression, using the definable imperative constructs, is:
$$
\hspace{-2mm}\begin{array}{l}
\alpha_{\rm SG}(\underline{ E},\underline{ \Root}, \underline{ A},\underline{ B})   :=  \ M_{base\_case};\\
\ \    {\bf repeat \ } 
\    M_{ind\_case}\  \comp\  \Copy ; \\
\ \ {\ \bf until } \ \sigma^r_{\Reach_A \eq {\Root}}(\sigma^r_{\Reach_B \eq {\Root}}(\id)).
\end{array}
$$
While this expression looks like an imperative program, it really is a \emph{constraint} on all possible sequences of choices. 
The answer to the question
$
\strA \models_T | \alpha_{\rm SG} \rangle \Last
$
is yes 
iff $A$ and $B$ belong to the same generation in the tree.

\vspace{2mm}

\subsection{Linear Equations mod 2}

\noindent \fbox{\parbox{\dimexpr\linewidth-2\fboxsep-2\fboxrule\relax}{
		
		{Problem: \textbf{mod 2 Linear Equations   $\alpha_{\rm F}$}}

	Given: system $F$ of linear equations  mod 2 over vars $V$ given by two ternary relations $\Eq_0$ and $\Eq_1$ 
		
		Question: Is $F$ solvable? }}

\bigskip

We assume that $V^{\strA}$ is a set, and $\Eq_0^{\strA}$ and $\Eq_1^{\strA}$ are  relations, 
both given by an input structure $\strA$ with $dom(\strA)=V^{\strA}$. Intuitively, $V^{\strA}$ is a set of variables, and $(v_1,v_2,v_3) \in \Eq_0^{\strA}$ iff $v_1\oplus v_2\oplus v_3 =0$, and $(v_1,v_2,v_3) \in \Eq_1^{\strA}$ iff $v_1\oplus v_2\oplus v_3 =1$. Such systems of equations are an example of constraint satisfaction problem that is not solvable by $k$-local consistency checks. This problem is known to be closely connected to the construction by Cai et al. \cite{CFI92}, and is not expressible in infinitary counting logic, as shown by Atserias, Bulatov and Dawar  \cite{AtseriasBD09}.  Yet, the problem is solvable in polynomial time by Gaussian elimination. We use the dynamic $\e$ operator to arbitrarily pick \emph{both} an equation (a tuple in one of the relations) and a variable (a domain element), on which Gaussian elimination $\Elim$ is performed.
$$
\begin{array}{l}
\alpha_{\rm F}(\underline{{ Eq_0}},\underline{{ Eq_1}},\underline{{ V}})   :=  \ M_{base\_case};\\
\ \   {\bf repeat \ } 
\    \Pick\_\Eq\_V \  \comp\  \Elim \  

{\ \bf until } \ \rneg \Pick\_\Eq\_V.
\end{array}
$$
Then, we have 
$\strA \models_T | \alpha_{\rm F} \rangle \Last
$\ \ iff\ \  $F$ is solvable.

\subsection{Observations}
In Reachability examples (s-t-Connectivity, Same Generation), propagations follow \emph{tuples} of domain elements given by the input structure, from one element to another. 
In Counting examples (Size 4, or any fixed size, Same Size, EVEN),
propagations are made arbitrarily, they are \emph{unconstrained}. 
In Mixed examples (mod 2 equations, CFI graphs),
propagations are  \emph{of both kinds}.
We believe that this is the reason of why adding \emph{just} counting to FO(FP) is not enough to represent all properties in PTIME \cite{CFI92} 
--- e.g., the algorithm for mod 2 Equations needs to interleave constrained and unconstrained propagations. Counting by itself cannot accomplish it.

\subsection{Mimicking Regular Expressions}
We now show that our logic can model the constructs  of Regular Expressions: the Kleene Star (also known as iteration or reflexive transitive closure) and the Union operation (called ``non-deterministic choice''  in programming language theory literature). The key idea  is to exploit the atomic non-determinism.
We assume that the input structure  contains a secondary domain, a set of elements that  serve as names of process expressions. This domain is
represented by a unary relation that interprets  relational variable $\Process$. This assumption is necessary for $\e$ to mimic non-determinism of Regular expressions.
$\ChooseProcess$ is a non-deterministic atomic transition defined as follows.
$$
\begin{array}{l}
\ChooseProcess \  : = \\
{\ \ \ \ \ \ } \e\defin{ \ChoosenProcess(x ) \rul \underline{\Process(x)}}.
\end{array}
$$
The Kleene Star operation is represented as 
$$
\begin{array}{l}
KleeneStar(\Process_1) \ := \ \\
\textbf{while} \ \sigma_{\ChoosenProcess = \Process_1}(\id)\\
\textbf{do} \ \ChooseProcess \comp \ExecuteProcess_1.
\end{array}
$$
The Union operation of Regular Expressions is expressed as follows.
$$
\begin{array}{l}
\Union(\Process_1,\Process_2) \ := \ \\
 \ \ChooseProcess \comp (\ExecuteProcess_1 \sqcup \ExecuteProcess_2).
\end{array}
$$
Here, $\ExecuteProcess_i$ is a modification of the $i$-th process expression, where, in addition to its operations, we carry its name. The name is represented by a randomly assigned element of the secondary domain of the input structure, i.e., the interpretation of the relational variable $\Process$.

Notice that, unlike the previous examples of polynomial-time Reachability and Cardinality properties, non-determinism here is used in a qualitatively different  way.  Using \emph{one} arbitrary, but fixed,  Choice function, as in Theorem \ref{th:membership-PTIME}, would not give us the full power of Regular Expressions.

\vspace{2ex}

\section{PTIME Turing Machines}\label{sec:PTIME-TM}

In this section, we demonstrate that our logic is strong enough to axiomatize any PTIME Turing machine over unordered structures. More specifically, we show that, for every polynomial-time recognizable class $\cK$ of structures, there is a sentence of logic $\Lo$ (introduced in Definition \ref{def:main-logic}) whose models are exactly $\cK$.
This corresponds to part (ii) of Definition \ref{def:logic-captures-PRTIME} of a logic for PTIME.

 We focus on the boolean query 	
$
   	\strA  \models_T |\alpha_{\rm TM} \rangle \Last	
$
 (see (\ref{eq:models-last})) and outline such a construction. 
The main idea is that a linear order on domain elements  is \emph{induced} by a path in a Kripke structure, that is, by a sequence of choices. In this path, we produce new elements one by one, as in the EVEN example. The linear order corresponds to an order on the tape of a Turing machine. After such an order is guessed, a deterministic computation, following the path, proceeds for that specific order. 
Note that, similarly to Theorem \ref{th:membership-PTIME} and unlike our construction for Regular Expressions, \emph{one arbitrary Choice function is enough to generate a specific order} on a tape of a deterministic  Turing machine. 
We now explain this construction.

We assume, without loss of generality that the machine runs for $n^k$ steps, where $n$ is the size of the domain. The program is of the form:
$$
\alpha_{\rm TM} (\strA):= 
 {\rm ORDER} \ \comp \ {\rm START}\  \comp \
  {\bf repeat \ } 
   {\rm STEP}  {\ \bf until } \ {\rm END}.
$$ 
\ignore{
$$
\begin{array}{ll}
\alpha_{\rm TM} (\strA):= & 
{\rm ORDER} \ \comp \ {\rm START}  \comp \\
&  \ \ \ \  \ \ \ \  \ \ \ \ \ \ {\bf repeat \ } \\
&  \ \ \ \   \ \ \ \   \ \ \ \  \ \ \ \  \ \ \ \  {\rm STEP} \\
&  \ \ \ \  \ \ \  \ \  \ \ \ \ {\ \bf until } \ {\rm END}.
\end{array} 
$$ 
} 
Procedure ORDER: Guessing an order is perhaps the most important part of our construction. This is the only place where $\varepsilon$ is needed.
We use a secondary numeric domain with a linear ordering. We guess elements one-by-one, as in the EVEN example, and associate an element of 
the primary domain with an element of the secondary one, using co-existence in the same structure.  Each path corresponds to a possible linear ordering.\footnote {Note that the order is defined \emph{not} with respect to an input Tarski structure, but with respect to the Kripke structure where  Tarski structures are domain elements.}  Once an ordering is produced, the primary domain elements are no longer needed (for search problems, $\BE$ modality can  be used  to reconstruct the correspondence.
Later, in START procedure and in the main loop, we use $k$-tuples of the elements of the  secondary domain for positions on the tape and time,   to count the steps in the computation. We use the lexicographic ordering on these tuples. 
The order on the elements of the auxiliary secondary domain is a \emph{binary} relation, but it is given on the \emph{input} (all non-input relations are unary). Thus, at any time of the computation, we can check that relation. Using the property of LLIF that, once we are on a path, we stay on it, ``next in the lexicographic ordering'' relation can be produced on the fly, using tuples of  unary definable relations, that change from state to state in the Register Kripke structure $\cR_{\strA}$, see Definition \ref{def:Register-structure}.

Procedure START: 
This procedure creates  an encoding of the input structure $\strA$ (say, a graph)
in a sequence of structures in the transition system, 
to mimic an encoding $enc(\strA)$ on a tape of a Turing machine. We use structures to represent cells of the tape of the Turing machine (one $\tau$-structure = one cell).  
The procedure follows a specific path, and thus a specific order  generated by the procedure START. An element is smaller in that order if it was generated earlier, which can be checked using  nested $\BE$ modalities.
Subprocedure 
$
{Encode}(\underline{vocab(\strA)}, \dots, S_\sigma, \dots , \bar{P}, \dots)
$
operates as follows. In every state (= cell), it keeps the input structure $\strA$ itself, and adds the part of the encoding of $\strA$ that belongs to that cell.
The interpretations of $\bP$ over the secondary domain of labels provide cell positions on the simulated input tape. 
Each particular encoding is done for a specific induced order on domain elements, in the path that is being followed. 

In addition to producing an encoding, the procedure START sets the  state   of the Turing machine to be
the initial state $Q_0$.  
It also sets initial values for the variables used to axiomatize the instructions of the Turing machine. 

Expression START is similar to the first-order formula  $\beta_{\sigma}(\bar{a})$ used by  Gr\"{a}del in his proof of capturing PTIME using SO-HORN logic on ordered structures \cite{Graedel:1991}. 
The main difference is that instead of tuples of domain elements $\bar{a}$ used to refer to the addresses of the cells on a tape, 
we use tuples $\bP$, also of length $k$.
Gr\"{a}del's formula $\beta_{\sigma}(\bar{a})$  for encoding input structures has the  following property:
$$
(\strA,<) \models \beta_{\sigma}(\bar{a}) \ \  \Leftrightarrow\ \  \mbox{ the $\bar{a}$-th symbol of $enc(\strA)$ is $\sigma$.} 
$$
Here, we have:  
$$
\begin{array}{c}
\ppath(\strA),0  \models_T Encode(\dots, S_\sigma, \dots , P_1(a_1), \dots, P_k(a_k), \dots ) \\ \Leftrightarrow\ \  \mbox{ the $P_1(a_1), \dots, P_k(a_k)$-th symbol of $enc(\strA)$ is $\sigma$,}
\end{array} 
$$
where $\ba$ is a tuple of elements of the secondary domain, $\ppath(\strA) $ is a path of the form (\ref{eq:string}) that starts in input structure $\strA$ and 
 induces a linear order on the input domain through an order on  structures
(states in the Register transition system $\cR_{\strA}$ rooted in $\strA$).  That specific generated  order is used in the encoding of the input structure. Another path produces a different order, and constructs an encoding for that order.

Procedure STEP: This procedure encodes the instructions of the deterministic Turing machine. Output-Monadic singleton-set  restrictions of Conjunctive Queries are well-suited for this purpose. Instead of time and tape positions as arguments of binary predicates as in Fagin's \cite{Fagin:1974} and Gr\"{a}del's \cite{Graedel:1991} proofs, we use coexistence 
in the same structure with $k$-tuples of domain elements, as well as lexicographic successor and predecessor on such tuples.\footnote{Recall how we used coexistence in a state in the Same Generation example.}
Polynomial time of the computation is guaranteed because 
time, in the repeat-until part, is tracked with $k$-tuples of domain elements.

Procedure END: This part checks if the accepting state of the Turing machine is reached.

We have that, for any PTIME Turing machine, we can construct term $\alpha_{\rm TM}$ in logic $\Lo$ introduced in Definition \ref{def:main-logic}, such that 
 $
   	\strA  \models_T |\alpha_{\rm TM} \rangle \Last	
$
if and only if the Turing machine accepts an encoding of $\strA$ for  some specific but arbitrary order of domain elements on its tape.

\vspace{2ex}

\section{Discussion}
\label{sec:Discussion}

In this section, we summarize our contributions, and discuss what needs to be done next, toward answering the main question. We also outline other future research directions.

\subsection{Summary and the Next Question} \label{sec:summary}
Motivated by the central quest for a logic for PTIME in Descriptive Complexity, we have defined  a Deterministic Dynamic Logic that can refer back in the executions. 
The logic is, simultaneously, an algebra of binary relations and a modal Dynamic logic over finite traces (here, called Two-Way DetLDL$_f$). While the operations  are rather restrictive (they are function-preserving), the logic defines the main constructs of Imperative programming. 
We use a \emph{dynamic} version of Hilbert's Choice operator Epsilon that, given a history, chooses a state, i.e.,  a domain element of a Kripke structure. This dependency on the history is crucial for formalizing, e.g., counting.
To provide a formal language perspective, we defined an equivalent Path semantics of DetLDL$_f$. The semantics associates, with each algebraic term $\alpha$,  a formal language $\cL(\alpha)$.

Atomic transductions are defined to be a singleton-set restriction of Monadic Conjuctive Queries.  These queries, intuitively, represent non-deterministic conditional assignments. Together with  algebraic terms used as control, this gives us a machine model.
We have demonstrated that any PTIME Turing machine can be simulated. 
 
Both reachability and counting properties on unordered structures can be expressed in the logic,   even though it does not have a special cardinality construct. 
Several examples in Section \ref{sec:Examples} illustrate these properties, including Linear Equations mod two. We have also shown how Regular Expressions can be mimicked using Choice functions. 

For a fixed Choice function, we have proven that the main task is in PTIME.  However, there could be, in general, exponentially many such concrete Choice functions.  Thus, it is important to understand \emph{under what syntactic conditions on terms $\alpha$, evaluating the main query
$$\strA \models_T\  |\alpha\rangle \Last\ [\e]$$ 
can be done naively}.
Naive evaluations follow the rules of Structural Operational Semantics from the proof of Theorem \ref{th:membership-PTIME} in Section \ref{sec:proof-membership-in-PTIME}, while treating the function variable $\e$ simply as an \emph{arbitrary constant}.  When such a naive evaluation is possible, the query is  clearly choice-invariant. 
The logic from Definition \ref{def:logic-captures-PRTIME} would be  the logic $\Lo$ introduced in Definition \ref{def:main-logic}, plus these decidable conditions.

\subsection{Outlook}

The need for the conditions  mentioned in Section \ref{sec:summary} above leads us to the crucial question: what are the advantages of our proposal over previous approaches to capturing PTIME? Our answer is that we simply have an algebra,
$$
 \begin{array}{c}	
\ \ \ 	\alpha  :: =   
\id  \mid \varepsilon M (\bZ)  \mid    \rneg\alpha  \mid    \lneg \alpha \mid  
\alpha\comp \alpha  \mid \\
\alpha  \sqcup \alpha \mid    \sigma^r_{X=Y} (\alpha) \mid    \alpha ^\iter  \mid   \alpha^\converse,
\end{array}
$$ 
and algebraic techniques have proven to be useful in the related CSP area, in the break-through results by Bulatov and Zhuk \cite{Bulatov17,Zhuk17}. 
Moreover, the situation, in our case, is  more transparent than in the CSP case. This is probably because, from the very beginning, we have   designed this algebra  with the desired condition in mind. The condition is \textbf{invertibility of algebraic terms}.\footnote{Our inspiration came from the well-known fact that reversible computations (such as those on Toffoli gates) do not consume any energy, and thus are the most efficient.} We identified invertibility as a condition for excluding non-determinism. Moreover, we believe that, for any given term in logic $\Lo$, this condition is decidable, and there are several ways of deciding it.

\subsection{Future Directions}
As one of the future directions, we would like to develop an automata-theoretic counterpart of DetLDL$_f$. With each  formula $\alpha$ in (\ref{eq:LLIF-two-sorted}), we will associate  an alternating 
automaton on finite words (AFW) $\cA_{\alpha}$.
A standard approach is then to transform the AFW to a non-deterministic finite automaton (NFA).
 However,  since we restricted all  terms to be deterministic, the  NFA is, in fact, a deterministic finite automaton (DFA).
An automaton for the full logic (\ref{eq:LLIF-two-sorted}),  needs to ``look inside the alphabet'' of $\tau$-structures $\Un$ to account for Selection, and to ``look backwards in time'' to account for the backwards-facing Unary Negation 
and Converse.

Another direction is to establish connections to non-classical logics.
We restricted classical connectives $\lor$ and $\neg$  to their deterministic counterparts, Preferential Union and two Unary Negations, respectively.   
Unary Negation has properties of Intuitionistic Negation. It is not idempotent, $\rneg \rneg \alpha \neq \alpha$, but is weakly idempotent, $\rneg \rneg \rneg \!\! \alpha = \rneg \alpha$.
 Moreover, $|\alpha] \phi$ seems similar to intuitionistic implication $\alpha \to \phi$. However, unlike intuitionistic logic, our semantics is parameterized with a  history-dependent Choice function, which leads to the duality of our box and diamond modalities. Moreover, we have argued that the dependency of the Choice functions on the past history is important for formalizing both programs and proofs in a proof system. 
 It would be interesting to investigate further connections to  intuitionistic and substructural logics, and well as to other logics of constructive reasoning.

\section{Acknowledgements}
\label{sec:Acknowledgements}
I am grateful to several colleagues for useful discussions,  and to Platform 7 coffee shop for their Secret Garden, a wonderful place to work, and an infinite supply of espresso. This paper was written, in part, while the author was a long-term visitor of the Simons Institute for the Theory of Computing in the Spring of 2021.

\bibliographystyle{alpha}
\bibliography{bibliography}

\end{document}